\documentclass[11pt]{article}

\usepackage{amsmath,amsfonts,amsthm,amssymb,color,bbm}
\usepackage{fancybox}
\usepackage{algorithmic}
\usepackage{algorithm2e}
\usepackage{dsfont}
\usepackage{bm}
\usepackage{cleveref}

\usepackage[margin=1in]{geometry}

\newcommand{\set}[1]{\left\{#1\right\}}

\newtheorem{theorem}{Theorem}
\newtheorem{corollary}[theorem]{Corollary}
\newtheorem{example}[theorem]{Example}
\newtheorem{lemma}[theorem]{Lemma}

\newtheorem{proposition}[theorem]{Proposition}

\newtheorem{fact}[theorem]{Fact}

\newtheorem{problem}[theorem]{Problem}

\theoremstyle{definition}
\newtheorem{definition}[theorem]{Definition}
\newtheorem{remark}[theorem]{Remark}

\newenvironment{fminipage}%
  {\begin{Sbox}\begin{minipage}}%
  {\end{minipage}\end{Sbox}\fbox{\TheSbox}}

\def\expec#1#2{{\mathbb{E}}_{#1}\left[ #2 \right]}
\newcommand{\var}[2][]{\mbox{\bf Var}_{#1} \left[ #2 \right] }

\newcommand*\diff[1][]{\mathop{}\!\mathrm{d^{#1}}}

\def\defeq{\stackrel{\mathrm{def}}{=}}
\def\setof#1{\left\{#1  \right\}}

\def\dim#1{\mathrm{dim} (#1)}

\def\norm#1{\left\| #1 \right\|}

\newcommand\GG{\boldsymbol{\mathit{G}}}

\usepackage{pgf,tikz}
\usetikzlibrary{arrows}

\newcommand\Tr{\text{\textbf{Tr}}}

\renewcommand\AA{\boldsymbol{\mathit{A}}}
\newcommand\BB{\boldsymbol{\mathit{B}}}
\newcommand\bb{\boldsymbol{\mathit{b}}}
\newcommand\dd{\boldsymbol{\mathit{d}}}
\newcommand\LL{\boldsymbol{\mathit{L}}}
\newcommand\XX{\boldsymbol{\mathit{X}}}
\newcommand\ZZ{\boldsymbol{\mathit{Z}}}

\newcommand\YY{\boldsymbol{\mathit{Y}}}
\newcommand\yy{\boldsymbol{\mathit{y}}}
\newcommand\zz{\boldsymbol{\mathit{z}}}
\newcommand\xx{\boldsymbol{\mathit{x}}}
\newcommand\uu{\boldsymbol{\mathit{u}}}
\newcommand\DD{\boldsymbol{\mathit{DD}}}
\newcommand\ee{\boldsymbol{\mathit{e}}}
\newcommand\diag{\text{diag}}

\newcommand{\inprod}[1]{\left\langle #1 \right\rangle}

\newcommand{\spcon}{\mathfrak{C}}

\newcommand{\exclude}[1]{}

\newcommand\rank{\text{rank}}

\newcommand\MVED[1]{\text{\textbf{maximum variance embedding}}_{#1}}

\def\Span#1{\textbf{Span}\left(#1  \right)}

\def\Aff#1{\textbf{Aff}\left(#1  \right)}

\newcommand\zero{\mathbf{0}}

\begin{document}

\title{$\lambda_\infty$ \textit{\&} Maximum Variance Embedding: Measuring and Optimizing Connectivity of A Graph Metric}

\author{
    Majid Farhadi\thanks{Georgia Institute of Technology.
\texttt{farhadi@gatech.edu}. Supported in part by the ACO Ph.D. Program.}
    \and
    Anand Louis\thanks{Indian Institute of Science \texttt{anandl@iisc.ac.in}. Supported in part by SERB Award ECR/2017/003296 and a Pratiksha Trust Young Investigator Award.}
    \and
    Mohit Singh\thanks{Georgia Institute of Technology \texttt{mohit.singh@isye.gatech.edu}.}
    \and
    Prasad Tetali\thanks{Georgia Institute of Technology.
\texttt{tetali@math.gatech.edu}. Supported in part by the NSF grants DMS-1811935 and NSF TRIPODS-1740776.}
}

\date{\today}

\maketitle

\begin{abstract}
	Bobkov, Houdr\'e, and the last author \cite{BHT00} introduced a Poincar\'e-type functional parameter, $\lambda_\infty$, of a graph $G = (V=[n], E)$. They related $\lambda_\infty$ to the {\em vertex expansion} of the graph via a Cheeger-type inequality, analogous to the  inequality relating the spectral gap of the graph, $\lambda_2$, to its {\em edge expansion}.
	While $\lambda_2$ can be computed efficiently, our understanding of computational complexity of $\lambda_\infty$ is relatively little. A work by the second author, Raghavendra, and Vempala~\cite{LRV13} related the complexity of $\lambda_\infty$ to the so-called small-set expansion (SSE) problem, and set forth the desiderata for NP-hardness of this optimization problem. We confirm the conjecture that computing $\lambda_\infty$ is NP-hard, for weighted trees. 
	
	\medskip
	
	Beyond measuring the connectivity of a graph by parameters such as $\lambda_2$ and $\lambda_\infty$, in many applications the objective is to optimize connectivity. 
	This, via convex duality, leads to a problem in machine learning, i.e., Maximum Variance Embedding (MVE). The output of MVE is a function from vertices of the input graph $G = (V = [n], E)$ into a low dimensional Euclidean space. The objective is to maximize the variance of the embedding, i.e., sum of pairwise squared Euclidean distances; and the constraints are specific upper, lower, or equality bounds on Euclidean distances between the neighboring vertices. 
	Special cases of Maximum Variance Embedding into $\mathbb{R}^n$ and $\mathbb{R}^1$ lead to \emph{absolute algebraic connectivity} \cite{F90} and \emph{spread constant} \cite{ABS98} of the graph metric, that respectively embrace well-connectivity of $G$ and $G^n$. Beyond connectivity, Maximum Variance Embedding has applications in diffusion speed and robustness of networks, clustering, and dimension reduction.
	
	\medskip
	
	We show that computing Maximum Variance Embedding in tree-width (of the graph) dimensions is NP-hard; accompanied by an efficient \emph{distance preserving} dimension reduction algorithm with no loss to the variance, given a tree-decomposition and allowing two additional dimensions.
	We show that MVE of a tree in $\mathbb{R}^2$ defines a non-convex yet benign optimization landscape, i.e., any local optima is a global optimum. 
	Moreover, we develop a combinatorial algorithm to find such optima in linear time. 
	Finally, we discuss approximate Maximum Variance Embedding being tractable in significantly lower dimensions.
	For trees and general graphs, for which Maximum Variance Embedding cannot be solved in less than $2$ and $\Omega(n)$ dimensions, unless $\text{P} = \text{NP}$, we provide $1+\varepsilon$ approximation algorithms for embedding into $1$ and $O(\log n /\varepsilon^2)$ dimensions, respectively.
\end{abstract}

\newpage

\tableofcontents

\newpage

\section{Introduction}
Measuring and optimizing connectivity of a network, modeled as a graph, has far-reaching applications, from enhancing security of data networks and communication systems to improving robustness and control of supply chain and pandemics, to name but a few. Connectivity can be measured through various quantities, such as minimum cut that can be found in polynomial time. Normalizing size of the cut by the significance of the resulting disconnectivity, e.g., volume of the smallest (dis)connected component, leads to \emph{isoperimetric} measures such as edge/vertex expansion that embrace broader theoretical and practical applications, e.g., see \cite{M94, HLW06, FCMR08, GS12, K16, ARV09, ZLM13} (and references therein), yet are hard to compute.

\emph{Spectral gap} of the graph, denoted by $\lambda_2$, can be characterized as the second smallest eigenvalue of the normalized Laplacian matrix of the graph (see \cite{CG97}) and efficiently computed. $\lambda_2$ estimates various isoperimetric constants through Cheeger inequalities, e.g.,
$$
\lambda_2/2 \leq \phi^E(G) \leq \sqrt{2\lambda_2}\,,
$$
where $\phi^E(G)$ is the {\em edge expansion} of the undirected graph $G = (V,E)$.

Bobkov, Houdr\'e, and the last author \cite{BHT00} introduced a novel Poincar\'e-type functional parameter, $\lambda_\infty\,,$ and derived new Cheeger-type inequalities, e.g.,
\begin{equation}
\label{linfcheeger}
    \lambda_\infty/2 \leq \phi^V(G) \leq 4 \lambda_\infty + 4 \sqrt{\lambda_\infty}\,,
\end{equation}
where $\phi^V(G)= \phi^V(G, \pi)$ denotes the {\em vertex expansion} of the graph $G$ for an arbitrary probability measure over the vertex set $\pi\,$. Formally, vertex expansion is defined as
$$
\phi^V(G) \defeq \min_{S \subseteq V} \frac{\pi(N(S) \cup N(V \setminus S))}{\min\{\pi(S),\pi(V \setminus S)\}}\,,
$$
where for $S \subseteq V$, $N(S) \defeq \{j \in V \setminus S: \exists i \in S \textrm{ such that } \{i,j \} \in E \}.$

For an undirected graph $G=(V,E)$ with non-negative edge weights given by $w: E \to \mathbb{R}^{\geq 0}$,
$\lambda_2$ can be defined as 
$$
\lambda_2 \defeq \inf_{f: V \to \mathbb{R}} \frac{ \expec{v \sim \pi}{\expec{u \sim \sigma_{N(v)}} {|f(v)-f(u)|^2}}} {\var[\pi]{f}}\,,
$$
where $\pi$ is the stationary probability distribution over the vertices and $\sigma_{N(v)}$ is a probability distribution over the vertices in $N(v)$ where a vertex $u \in N(v)$ is sampled with probability proportional to $w(u,v)$.

In contrast, Bobkov, Houdr\'e, and the last author   defined $\lambda_\infty$ as
\begin{equation}
    \label{def:linf}
 \lambda_\infty \defeq \inf_{f: V \to \mathbb{R}}
\frac{\expec{v \sim \pi}{ \sup_{u \in N(v)} |f(v)-f(u)|^2}}{\var[\pi]{f}}, 
\end{equation}
where $\pi$ is an arbitrary probability distribution over the vertices.
In general one can convert bounds for vertex and edge expansion at the cost of a degrade by a multiplicative factor of maximum degree of the graph, $\Delta = \Delta(G)$. However, directly using (an estimate for) $\lambda_\infty$ can yield better bounds on vertex expansion and relevant measures (\Cref{linfcheeger}), as pointed out by \cite{BHT00} in the context of refining certain isoperimetric and concentration bounds from \cite{AM85}.
As another application, $\lambda_2$ can help bounding the mixing rate of a Markov chain while $\lambda_\infty$ is connected to a different dispersion process \cite{CLTZ18}.

Earlier, Fiedler \cite{F73} coined algebraic connectivity of the graph as the second smallest eigenvalue of its Laplacian matrix, i.e., equivalent to $\lambda_2$.
The maximum such quantity that one can achieve, using a total weight of $1$ assigned to edges of the graph, is called \emph{absolute algebraic connectivity} of the graph
\cite{F90, F93}.
Absolute algebraic connectivity can be inversely formulated as the minimum required total weight of the edges of the graph, to ensure the second smallest eigenvalue of the Laplacian is greater than or equal to one. 

Absolute algebraic connectivity can be formulated as a semidefinite program, later formulated as Eq.~\eqref{dualsdp}.
The dual semidefinite program, Eq.~\eqref{eq:primal-sdp}, has a further geometric interpretation.
That is, the vertices of the graph are to be embedded into $\mathbb{R}^n$ such that the neighboring vertices are within/at unit distance from each other, and the objective is to maximize the Euclidean variance of the embedding.

In machine learning literature, Maximum Variance Embedding (MVE) into arbitrary dimensions has been a proxy for non-linear dimension reduction, while preserving local structure of the underlying manifold, that is embraced by a proximity graph \cite{WS06aaai}. Our study also contributes to the theory of non-linear dimension reduction, that is an active area of research \cite{TJ18}.
Further applications of maximum variance embedding, and the dual problem, are in the design of fastest mixing Markov processes \cite{SBXD06}, optimizing robustness and restoration of power networks \cite{edstrom2011spectral}, convergence rate of quantum consensus \cite{jafarizadeh2016optimizing}, and control of multi-agent systems \cite{choi2017multi}. On growing applications/interests for maximizing algebraic connectivity of a network and related problems, refer to \cite{gomyou2020optimal, tavasoli2020maximizing} and references therein.

A special embedding problem, this time into $\mathbb{R}^1$, was 
first introduced by Alon, Boppana, and Spencer \cite{ABS98}. They showed that isoperimetric properties of Cartesian products of a (graph) metric \cite{CT98}, in asymptotic fashion, can be well-understood in a wide range 
by merely bounding the {\em second moment} over Lipschitz valuations of the vertices. Namely the spread constant of the graph, denoted by $\spcon = \spcon(G, \pi)\,,$ was defined as
$$
\sup_{f \in \mathcal{L}(G)} \var[\pi]{f}\,,
$$
where $\mathcal{L}(G)$ denotes the set of Lipschitz functions  $f: V \rightarrow \mathbb{R}$ with respect to the distance metric defined by $G\,,$ i.e., satisfying
$|f(u)-f(v)| \leq 1\,, \quad \forall \set{u,v} \in E\,.$
It is easy to see that the spread constant is (tightly) upper bounded by the inverse of $\lambda_\infty$, i.e.,  $\spcon \leq 1/\lambda_\infty\,.$

Let us formulate the Maximum Variance Embedding problem, that generalizes the aforementioned problems. Given a graph $G = (V,E)$, a positive integer $k$, along with a probability distribution $\pi: V \rightarrow \mathbb{R}^{\ge 0}$ over the vertex set and a set of (upper, lower, or equality) bounds on the Euclidean distances of the endpoints of the corresponding edges, $d: E \rightarrow \mathbb{R}^{\ge 0}$; the output should be an embedding of vertices in the $k$-dimensional Euclidean space, $\yy: V \rightarrow \mathbb{R}^k$, equivalently a matrix $\YY = [\yy_1 | \cdots | \yy_n]^T \in \mathbb{R}^{k \times n}$, that satisfies given in/equalities of form $\|\yy_i-\yy_j\|$ and $d_{ij}$, corresponding to every edge in $G$. That is, every $ij \in E$ imposes a lower bound, equality, or upper bound on the distance between its endpoints in the embedding. The objective is to maximize the variance of the embedding with respect to $\pi$,
\[ \max_{\yy: V \rightarrow \mathbb{R}^k} \var[\pi]{\yy} \text{ s.t. } \|\yy_i - \yy_j\| \leq | = | \ge d_{ij} \quad \forall ij \in E\,.
\]

The objective, i.e., variance, can be written as follows (derived in \Cref{ssec:variance}). 
\begin{align}
    \var[\pi]{\yy} \defeq \expec{i \sim \pi}{\|\yy_i - \expec{j \sim \pi}{\yy_j}\|^2} = \sum_{i < j} \|\yy_i - \yy_j\|^2 \pi_i \pi_j
= \langle \text{Diag}(\pi)- \pi \pi^T , \YY \YY^T \rangle
\end{align}
Denoting the Laplacian of a single edge $ij \in E$ by $\LL_{ij} = (e_i-e_j) (e_i-e_j)^T$, one can write $\|\yy_i - \yy_j\|^2 = \langle \LL_{ij}, \YY \YY^T \rangle$ where $\langle \AA,\BB \rangle = \Tr(\AA{}^T \BB)$. $\YY \YY^T$ is a positive semidefinite matrix of rank $k$. Maximum variance embedding is equivalent to the following rank-constrained Semidefinite program, generalizing problems from previous studies \cite{BDX04, WS04,MT06}.
\begin{align}
\label{eq:primal-sdp}
\max \langle \XX, \diag(\pi) \rangle \text{ s.t. } \langle \XX, \pi \pi^T \rangle = 0, \langle \XX, \LL_{ij} \rangle \le|=|\ge d_{ij}^2 \forall ij \in E, \XX \succcurlyeq \mathbf{0}, \rank{\XX} \leq k
\end{align}

The dual SDP (for $k = n$) can be written as
\begin{align}\label{dualsdp}
\sum_{ij \in E} w_{ij} d_{ij}^2
\text{ s.t. }
\sum_{ij \in E} w_{ij} \LL_{ij} + \mu \pi \pi^T \succcurlyeq \diag(\pi),
w_{ij} \in \mathbb{R}^{\ge 0}|\mathbb{R}|\mathbb{R}^{\leq 0} \quad \forall ij \in E, \mu \in \mathbb{R}\,.
\end{align}
For uniform $\pi$ and $d$ this would be the minimum total weight required to increase the second smallest eigenvalue of the Laplacian beyond $1$:
\[
\min \sum_{ij \in E} w_{ij}
\text{ s.t. }
\lambda_2(\sum_{ij \in E} \LL_{ij} \cdot w_{ij}) \ge 1, \quad w_{ij} \in \mathbb{R}^{\ge 0}|\mathbb{R}|\mathbb{R}^{\leq 0} \quad \forall ij \in E.
\]
Note that an upper bound constraint $\langle \XX, \LL_{ij} \rangle \leq d_{ij}^2$ in the primal, corresponds to $w_{ij} \ge 0$ in the dual, and a lower bound constraint corresponds to $w_{ij} \leq 0$. Fixing the length of an edge $ij \in E$ in the embedding, allows $w_{ij}$ to take any real value. 
For non-negative weights, we have (scaled inverse) of the absolute algebraic connectivity problem, while allowing negative weights may lead to further applications. 

When the upper bound on the rank of $\XX$ in \eqref{eq:primal-sdp} is trivial, i.e., embedding is in $k = n$ dimensions, we have an ordinary SDP that is a convex program and can be solved up to arbitrary precision in strongly polynomial time, e.g., using interior point methods \cite{NN94, Nes13}. Intuitively speaking, the MVE problem becomes harder in lower dimensions.

\subsection{Results \& Organization}

Previous studies \cite{BHT00,LRV13} imply that $\lambda_\infty$ is Small Set Expansion hard (which is a weaker hardness result than NP-hardness, see \cite{RS10}) to approximate better than $O(\log \Delta)$ in certain parameter regimes, (refer to \cite{LRV13} for formal statement). Even though the small-set
expansion hypothesis remains unproven, this suggested that the computation of $\lambda_\infty$ is likely to be hard. Nevertheless, the fundamental question of NP-hardness of the computation of
$\lambda_\infty$ remained unresolved since $2000\,.$ We settle this question in \Cref{sec:lambda-infinity} by proving that computing $\lambda_\infty$ is NP-hard, even in the case of weighted star graphs (see \Cref{thm:nph_theorem}). This is in significant contrast to $\lambda_2$, which can be efficiently computed for general graphs (see \Cref{polybits}).
Our technique also leads to an NP-hardness result for computing vertex expansion of a weighted tree (\Cref{thm:vertex-exp-hard}), in contrast to edge expansion which can be easily computed. We also include a dynamic programming algorithm that efficiently computes vertex expansion of unweighted trees (\Cref{thm:vertex-exp-tree}).

For spread constant, i.e., $k = 1$, the SDP relaxation was shown to have an integrality gap $\Omega(\log n / (\log \log n))$ \cite{Naor14}, and has led to approximations with multiplicative error no more than $O(\log n)$ \cite{MT06,SBXD06}. However, no NP-hardness result for any of the maximum variance embedding problems was known, to the best of our knowledge.
In \Cref{sec:MVE-tree}, we show that spread constant is NP-hard to compute for (either edge or vertex) weighted trees (\Cref{thm:sc-nph_theorem}). This is complemented by providing a fully polynomial time approximation scheme for spread constant of weighted trees (\Cref{thm:fptas}).

Trees are the least structurally complex connected graphs. A quantitative measure of structural complexity is tree-width of the graph.
Generalizing our result for trees, in \Cref{sec:MVE-tree-width}, we show that finding maximum variance embedding in tree-width dimensions is NP-hard (\Cref{thm:hard}).
This is also a complement to the main result of  G\"oring, Helmberg, and Wappler \cite{GHW08}, who showed given an optimal maximum variance embedding in $n$ dimensions, with only upper bound inequality constraints, along with a tree-decomposition, an optimal maximum variance embedding in $1+$width (of the tree-decomposition) dimensions can be found in polynomial time.
We also notice that their algorithm fails for maximum variance embedding with equality constraints, due to a characterization of optimal embedding called separator-shadow property, that does not hold\footnote{
For example, let $V = [6]$ and $E = \{12, 23, 31, 45, 56, 64, 14\}$, i.e., a disjoint union of two 3-cycles with an additional edge in between. For any optimal embedding in $\mathbb{R}^{\ge 2}$, vertex $1$ is a separator that has neither of remaining components in its shadow (from the mean/origin).} even given arbitrary additional dimensions beyond tree-width.
For general maximum variance embedding, we present a new algorithm that using one more dimension can reduce the dimension of any feasible solution with no degrade to the objective function (variance), while preserving (Euclidean) distances between all adjacent nodes.

Another takeaway from this study is that while tree-width is a tight bound for existence and computability of maximum variance embedding, approximate MVE is feasible in significantly lower dimensions. For trees where maximum variance embedding can be found in polynomial time in no less than $2$ dimensions (for which we also provide a linear time combinatorial algorithm in \Cref{sec:MVE-tree}) unless $P = NP$, we design a $1+\varepsilon$ approximation algorithm for embedding into $\mathbb{R}^1$ (improving upon an immediate $2$-approximation due to previous studies).
The gap between dimensions where maximum variance embedding can be computed vs approximated is even more significant for general graphs, i.e., $\Theta(n)$ vs $\Theta(\log n)$. In \Cref{sec:MVE-general} we show that maximum variance embedding can be $1+\varepsilon$ approximated in $O(\log n / \varepsilon^2)$ dimensions (see \Cref{thm:JL}) generalizing techniques of \cite{MT06,SBXD06}.

\section{$\lambda_\infty$ of a Star}
\label{sec:lambda-infinity}
This Section is dedicated to complexity of computing $\lambda_\infty$ of a star.
In \S{}\ref{ssec:lambda-infty-NPH} we show the problem is (weakly) NP-hard, and in \S{}\ref{ssec:lambda-infty-star-approx} we present a $1+\varepsilon$ approximation algorithm for it.
We denote a star graph with $n-1$ leaves by $S_n$. Name the center $0 \in V$ which is connected to all other vertices $[n-1] \subseteq V$ that are leaves of our star (tree). The edge set of $S_n$ is $E = \set{ \set{0,i} | i \in [n-1]}\,.$

\subsection{Weak NP-hardness}\label{ssec:lambda-infty-NPH}

In the rest of this subsection we prove the following result.

\begin{theorem} 
\label{thm:nph_theorem}
There exists a polynomial $p$ such that the following holds.
Given a star graph $G\,,$ and rational probability distribution $\pi$ over the vertices, computing $\lambda_\infty(G,\pi)$ to $\text{p}(|G,\pi|)$ many bits accuracy is NP-hard (here $|G,\pi|$ denotes the length of the input string).
\end{theorem}

First, note that the infimum in the definition of $\lambda_\infty$ is bounded and feasible for any connected graph; which is explained by \Cref{lem:infmin} in \Cref{sec:omittedproofs} for completeness. Therefore, $\lambda_\infty$ is the optimal solution to the following optimization problem:
\begin{align}
\label{def:lambdainfty}
\lambda_\infty = \min_{x: V \to \mathbb{R}}
\expec{v \sim \pi}{\max_{u \in N(v)}  |x_u-x_v|^2}
\quad \text{s.t.} \quad \var[\pi]{x} \ge 1\,.
\end{align}

\begin{remark}
\label{polybits}
We note that $\lambda_2$ is an irrational number in general (e.g., the eigenvalues of the normalized adjacency matrix of an $n$-cycle are $\set{\cos (2 \pi k / n) : k \in [n]}$), and $\lambda_\infty$ is also believed to be irrational in general. Therefore, they are not computable exactly by a Turing machine. Thus in \Cref{thm:nph_theorem}, we restrict ourselves to the problem of computing it to an accuracy of polynomially many bits.  
\end{remark}

To prove \Cref{thm:nph_theorem} we reduce the NP-complete {\sc Integer Partitioning problem} (\Cref{prob:partition}) to computing $\lambda_\infty$ of the star graph. 
We find a lower bound, $\lambda_\infty(S_n,\pi) \ge \frac{1}{1-\pi_0}$, that is satisfied with equality if and only if the leaves can be \emph{partitioned in half} with respect to $\pi$

Since we want to prove NP-hardness for the problem of computing $\lambda_\infty$ to a polynomially many bits of accuracy, it will not suffice to show that the value of $\lambda_\infty$ is simply different for the {\sc yes} and the {\sc no} instances of the {\sc Integer Partitioning problem}. We will show a non-trivial separation between the {\sc yes} and {\sc no} instances of the {\sc Integer Partitioning problem} corresponds to a measurable (polynomially deep bit) difference in  $\lambda_\infty$.

To be able to provide this gap, we exploit new structural characterizations of a solution (to $\lambda_\infty$ viewed as an optimization problem defined by Equation~\eqref{def:lambdainfty}) of the star graph.
Applying first and second order optimality conditions, we show that there is at most one non-extreme (see \Cref{starchar}) variable in the solution to the underlying optimization problem. Similar approach leads to our hardness results for vertex expansion and spread constant. Let us begin by the definition of {\sc Integer Partition} Problem.

\begin{problem}[{\sc Integer  Partition} Problem]
\label{prob:partition}
Given positive integers $\set{p_1, \cdots, p_n}$ we are to decide whether this multiset can be partitioned into two subsets of equal sum.
\end{problem}

It is easy to verify that the lower bound $\frac{1}{1-\pi_0}$ will be achieved if the leaves are split in half, by embedding them at equal distances from the center. In this case we have a \emph{binary} and \emph{balanced} embedding, formally defined as follows.

\begin{definition}
For a star graph, we call a valuation/embedding $x: V \rightarrow \mathbb{R}\,,$ {\em binary}, if all leaves are assigned numbers at the same absolute distance from the value assigned to the root, i.e.,
$$
|x_k-x_0| = \max_i |x_i - x_0|, \ \ \forall k \in [n-1]\,,
$$
and call it {\em balanced} around vertex $v$, if its first order moment with respect to $x_v$ is zero, i.e.,
$$
\sum_{k \in V} \pi_k (x_k - x_v) = 0 \iff \expec{u \sim \pi}{x_u} = x_v\,.
$$
\end{definition}

We show the values assigned to vertices by an \emph{optimal solution} for $\lambda_\infty(S_n,\pi)$ is \emph{almost} binary, i.e., all but at most one of the leaves are at equal distances from the center. This paves the way for showing the aforementioned inequality, yet another major step remains. We need to show a large enough increase on $\lambda_\infty$ when a balanced partition of the leaves is not feasible.
Before we start let us denote a useful fact.

\begin{fact} \label{fact:mediant}
Given positive numbers $a, b, c, d \in \mathbb{R}\,,$ and $\frac{a}{b} < \frac{c}{d}\,,$ we have
$
\frac{a}{b} < \frac{a+c}{b+d} < \frac{c}{d}\,,
$
i.e.,  {\em mediant} of two (positive) ratios falls in between. Consequently, $\frac{e}{f} < \frac{e-g}{f-g}$ for $e > f > g > 0\,.$
\end{fact}

Let us rewrite \eqref{def:lambdainfty}, i.e., the optimization problem defined by $\lambda_\infty$, for a star.
\begin{align}
\lambda_\infty(S_n, \pi) &= \min_{x \in \mathbb{R}^V \setminus \set{\mathbf{0}}} \frac{\expec{v \sim \pi}{\max_{u \in N(v)}  |x_u-x_v|^2} }{\var[\pi]{x} }  \nonumber\\
&= \min_{x \in \mathbb{R}^V \setminus \set{\mathbf{0}}, x_0 = 0} \frac{\expec{v \sim \pi}{\max_{u \in N(v)}  |x_u-x_v|^2} }{\var[\pi]{x} } && \text{uniform shift invariance} \nonumber \\
&= \min_{x \in \mathbb{R}^V \setminus \set{\mathbf{0}}, x_0 = 0} \frac{\pi_0 \cdot \max_{i \in [n-1]} x_i^2 + \sum_{i \in [n-1]} \pi_i x_i^2 }{\var[\pi]{x} } && \text{star structure} \nonumber \\
&= \min_{x \in \mathbb{R}^V\setminus \set{\mathbf{0}}, x_0 = 0} \frac{\pi_0 \cdot \max_{i \in [n-1]} x_i^2 + \expec{i \sim \pi}{x_i^2} }{\expec{i \sim \pi}{x_i^2} - (\expec{i \sim \pi}{x_i})^2}\,.\label{laminf_star1}
\end{align}

In the above, $\mathbf{0}$ denotes the all zero vector.
We are going to characterize the optimum valuation, namely
\begin{equation}
\label{star_y}
  y \in {\sc argmin}_{x \in \mathbb{R}^V\setminus \set{\mathbf{0}}, x_0 = 0} \frac{\pi_0 \cdot \max_{i \in [n-1]} x_i^2 + \expec{i \sim \pi}{x_i^2} }{\expec{i \sim \pi}{x_i^2} - (\expec{i \sim \pi}{x_i})^2}\,.
\end{equation}

\begin{lemma}
\label{starchar}
There exists at most one vertex $i \neq 0$ with
$
|y_i| \ne \max_k |y_k|\,.
$
\end{lemma}
\begin{proof}
Note that $y_0 = 0$.
We prove the claim by contradiction, assuming $\exists i, j > 0$, $i \neq j$ for which
\begin{align}y_i, y_j \in (-\max_k |y_k|, \max_k |y_k|)\,. \label{fallsin}\end{align}

Let us abuse the notation for $\lambda_\infty(G,\pi)$ when $G$ and $\pi$ are clear from the context, 
denoting $\lambda_\infty$ as a function from valuation of the vertices to the objective value of the optimization problem, i.e., $\lambda_\infty: \mathbb{R}^V \rightarrow \mathbb{R}, y \mapsto {\expec{v \sim \pi}{\max_{u \in N(v)}  |y_u-y_v|^2}}/{\var[\pi]{y}}$.
We find a contradiction by showing $\lambda_\infty(\cdot)$ can be further decreased, slightly moving from $y$ along the direction
$v = \pi_i e_j - \pi_j e_i\,,$
i.e., adding to one and decreasing another, while keeping the expectation intact. Namely for 
$y' = y + \delta v\,,$
we have
\[ 
\expec{k \sim \pi}{y'_k} = \sum_{k} \pi_k y'_k 
= \left(\sum_{k} \pi_k y_k \right) - \pi_i \pi_j \delta + \pi_j \pi_i \delta 
= \expec{k \sim \pi}{y_k}. \]

The assumption by Equation~\eqref{fallsin} guarantees that for sufficiently small $\delta\,,$ we have
$$\max_k y_k^2 = \max_k y_k'^2\,.$$

The only term in formulation of $\lambda_\infty(\cdot)$ as of \Cref{laminf_star1}, that is affected by changing $y$ to $y'$ is in the second moment $\gamma \defeq \expec{k \sim \pi}{y_k'^2} - \expec{k \sim \pi}{y_k^2}\,,$ for which we have
\begin{align*}
\expec{k \sim \pi}{y'^2_k} &= \expec{k \sim \pi}{y^2_k} + \pi_i (-2 y_i \pi_j \delta + (\pi_j \delta)^2) + \pi_j (2 y_j \pi_i \delta + (\pi_i \delta)^2) \\
&= \expec{k \sim \pi}{y^2_k} + 2 \delta \pi_i \pi_j (y_j - y_i) + \delta^2 (\pi_i \pi_j^2 + \pi_j \pi_i^2)  = \expec{k \sim \pi}{y^2_k} + \gamma\,.
\end{align*}

 We can  assure a $\gamma > 0$ by setting
$\delta = \varepsilon \cdot (\text{sign}(y_j-y_i) + \mathds{1}[y_j - y_i = 0])\,,$
for a small $\varepsilon > 0\,,$ where $\mathds{1}[\text{condition}]$ is the indicator function denoting whether the condition is satisfied, i.e., $\mathds{1}[\text{True}] = 1$ and $\mathds{1}[\text{False}] = 0$.

Computing $\lambda_\infty$ at $y'$ we have
\begin{align}
\lambda_\infty(y') &= \frac{\pi_0 \cdot \max_{i \in [n-1]} y_i'^2 + \expec{i \sim \pi}{y_i'^2} }{\expec{i \sim \pi}{y_i'^2} - (\expec{i \sim \pi}{y_i'})^2} \nonumber \\
&=  \frac{\pi_0 \cdot \max_{i \in [n-1]} y_i^2 + \expec{i \sim \pi}{y_i^2} + \gamma}{\expec{i \sim \pi}{y_i^2} - (\expec{i \sim \pi}{y_i})^2 + \gamma} \label{mediant}\\
&\in \left(\frac{\gamma}{\gamma}, \frac{\pi_0 \cdot \max_{i \in [n-1]} y_i^2 + \expec{i \sim \pi}{y_i^2}}{\expec{i \sim \pi}{y_i^2} - (\expec{i \sim \pi}{y_i})^2} \right) = (1,\lambda_\infty(y)) \,,\label{ininterval} && \text{\Cref{fact:mediant}}
\end{align}
where we applied \Cref{fact:mediant} for $a = b = \gamma\,,$ $c = \pi_0 \cdot \max_{i \in [n-1]} y_i^2 + \expec{i \sim \pi}{y_i^2}$ and $d = \expec{i \sim \pi}{y_i^2} - (\expec{i \sim \pi}{y_i})^2\,.$
\Cref{mediant} shows $\lambda_\infty(y')$ is a median of two fractions $c/d = \lambda_\infty(y)$ and $a/b = 1\,.$
Noticing $c/d > 1\,,$ the last line provides the desired contradiction $\lambda_\infty(y') < \lambda_\infty(y)\,.$ 
\end{proof}

It is now easy to prove the desired tight lower bound on $\lambda_\infty$ of a star.

\begin{lemma} \label{startheorem}
$\lambda_\infty \geq \frac{1}{1 - \pi_0}$ and the equality holds if and only if $y$ is $(\mathbf{i})$ binary, and $(\mathbf{ii})$ balanced around the root.
\end{lemma}

\begin{proof}
{
Recall $y_0 = 0$. Without loss of generality assume the following.
\begin{align*}
\max_k |y_k| &= 1 && \text{scale invariance} \\
\exists~{i \in [n-1]} ~\forall~{k \neq i}~ |y_k| &= 1    && \text{\Cref{starchar}} 
\end{align*} 
Let 
$\pi_- = \sum_{k \neq i} \pi_k \cdot \mathds{1}[y_k = -1]$ and $
\pi_+ = \sum_{k \neq i} \pi_k \cdot \mathds{1}[y_k = +1]\,.
$
We need to show
$$
\lambda_\infty = \frac{\pi_- \cdot (-1)^2 + \pi_+ \cdot (1)^2 + \pi_i y_i^2 + \pi_0 \cdot (1)^2}{\pi_- \cdot (-1)^2 + \pi_+ \cdot (1)^2 + \pi_i y_i^2 - (\pi_- \cdot (-1) + \pi_+ \cdot 1 + \pi_i y_i)^2} \geq \frac{1}{1-\pi_0}\,,
$$
and the equality holds if and only if the valuation is binary and balanced, in which case,
$y_i^2 = 1$,
and 
$
\pi_+ \cdot 1 + \pi_i \cdot y_i + \pi_- \cdot (-1) = 0\,.
$

We show equivalence between the target inequality
$\lambda_\infty \ge \frac{1}{1-\pi_0}$ and inequality
$$ \pi_0 \pi_i (1-y_i^2) + (\pi_+ - \pi_- + \pi_i y_i)^2  \ge 0\,,$$
that is trivial due to the fact that $y_i \in [0,1]\,.$ Re-writing the latter inequality as
\[\pi_0 \pi_i - \pi_0 \pi_i y_i^2 + \pi_+^2 + \pi_-^2 + \pi_i^2 y_i^2 - 2 \pi_+ \pi_- - 2 \pi_- \pi_i y_i + 2 \pi_+ \pi_i y_i \ge 0 \]
and applying $\pi_i = 1 - \pi_+ - \pi_- - \pi_0\,,$ we get 
\[ \pi_0(1 - \pi_+ - \pi_- - \pi_0) - \pi_0 \pi_i y_i^2 + \pi_+^2 
 + \pi_-^2 + \pi_i^2 y_i^2 - 2 \pi_+ \pi_- - 2 \pi_- \pi_i y_i + 2 \pi_+ \pi_i y_i  \ge 0 \]
Rearranging and adding terms to both sides we get
\begin{multline*}
\pi_+ + \pi_- + \pi_i y_i^2 + \pi_0 - \pi_0 \pi_+ - \pi_0 \pi_- - \pi_0 \pi_i y_i^2 - \pi_0^2 \ge \\
\pi_+ + \pi_- + \pi_i y_i^2 - \pi_+^2 - \pi_-^2 - \pi_i^2 y_i^2 - 2 \pi_+ \pi_i y_i + 2 \pi_+ \pi_- + 2 \pi_- \pi_i y_i.
\end{multline*}
Factorizing this expression, we get
\[ (\pi_+ + \pi_- + \pi_i y_i^2 + \pi_0 ) \cdot (1-\pi_0) \geq
(\pi_+ + \pi_- + \pi_i y_i^2)-(\pi_+ - \pi_- + \pi_i y_i)^2. \]
A final rearrangement give us 
\[ \lambda_\infty = \frac{\pi_+ + \pi_- + \pi_i y_i^2 + \pi_0 }{(\pi_+ + \pi_- + \pi_i y_i^2)-(\pi_+ - \pi_- + \pi_i y_i)^2} \geq \frac{1}{1-\pi_0}\,, \]
where derivations (holding in both directions) also hold in equality form, and so does for strict inequality. We also implicitly used the fact that $\pi_i \in (0,1), \forall i$ and that $\var[\pi]{x} > 0\,,$ particularly assuring the denominators are positive in the last inequality.
}
\end{proof}

\paragraph{The Reduction}

Given a decision oracle for  $\lambda_\infty \leq \beta$, we provide a polynomial reduction from the NP-hard integer {\sc Partition} problem \cite{K72}.

Given an instance for the {\sc Partition} problem, namely $P = \set{p_j: j \in [n-1]}$ define a star $S_n$, with vertex set $V = \{0, \cdots, n-1\}$, and $E = \{(0,j) | j \in [n-1] \}$, and a probability measure $\pi$ on the vertex set defined as follows. 
\[
    \pi_j= 
\begin{cases}
    \frac{\beta - 1}{\beta}, & j = 0\\
    \frac{1}{\beta} \cdot \frac{p_j}{\sum_k p_k},              & j \in [n-1]
\end{cases}
\]

To respond to {\sc Partition}, we forward the answer from the oracle on whether $\lambda_\infty(S_n, \pi) \leq \beta\,.$
Note that applying  \Cref{startheorem} we know if the answer to {\sc Partition} is {\sc yes} then $\lambda_\infty = \beta$, and otherwise it is larger. All remains is to lower-bound the increase in $\lambda_\infty$ due to the {\sc no}-{\sc Partition} scenario, such that it affects a polynomially deep bit of the answer. We do this in the following lemma that shows $\lambda_\infty > \beta + \Omega((\beta-1)/(\beta(\sum_i p_i)^{2}))$ in {\sc no}-{\sc Partition} case, increasing a digit (of $\lambda_\infty$) no deeper than $O(\log \sum_i p_i)$ that is polynomially bounded by the length of binary representation of the input to the {\sc Partition} problem, i.e., $\Theta(\sum_i \log p_i)$.

\begin{lemma}
\label{lemma:reduction}
$\lambda_\infty = \beta$ if and only if $P$ can be {\sc Partition}ed, otherwise $\lambda_\infty \ge \beta + \Omega\left(\frac{\beta-1}{ \beta (\sum_k p_k)^2}\right)\,.$ 
\end{lemma}

\begin{proof}
Applying  \Cref{startheorem} for the star, we have 
$\lambda_\infty = \beta$
for some normalized optimal valuation $y \in \mathbb{R}^V\,,$ if and only if $y_k = \pm1$ for every leaf $k$ and it is balanced w.r.t.\,$\pi$, i.e.,
\begin{align*}
&\sum_k \pi_k \cdot \mathds{1}[{y_k > 0}] - \sum_k \pi_k \cdot \mathds{1}[y_k < 0] = 0\,
\iff
&\sum_k p_k \cdot \mathds{1}[{y_k > 0}] = \sum_k p_k \cdot \mathds{1}[y_k < 0]\,,
\end{align*}
which is proof of a {\sc Partition} of $P$ with corresponding parts $\{k | y_k < 0\}$ and $\{k | y_k > 0\}\,.$

The argument is clearly reversible and given a valid {\sc Partition}, i.e.,  $S \subseteq [n-1]$ where
$
\sum_{j \in S} p_j = \sum_{j \notin S} p_j\,,
$
the following assignment for $y$ (upper) bounds $\lambda_\infty$ by $\frac{1}{1-\pi_0} = \beta\,.$
\[
    y_k= 
\begin{cases}
    0, & k = 0\\
    1, & k \in S\\
    -1,              & \text{otherwise}
\end{cases}
\]

Considering the  {\sc no-Partition} case, let $b^*$ be the optimal balance (minimum difference of sums) of $\pi$ over all partitions of the leaves, i.e., $\min_{S \subseteq [n-1]} |\sum_{j \in S} \pi_j - \sum_{j \in [n-1] \setminus S} \pi_j|$. Since the $p_j$'s are integers, we have that the optimal balance is at least
\begin{align}
b^* \ge \frac{1}{\beta} \cdot \frac{1}{\sum_j p_j}\,. \label{min-imbalance}
\end{align}
Let $y$ be the vector corresponding to $\lambda_\infty$ of this star instance. As before, we will assume that $\max_{k}|y_k| = 1\,,$ $y_0 = 0\,,$ 
and let $i \in [n-1]$ be the only (if any) index such that $|y_i| < 1$ (otherwise any index). By symmetry we can also assume $y_i \ge 0$.  
We are going to lower bound
\begin{align}
\lambda_\infty &= \frac{\pi_+ + \pi_- + \pi_i y_i^2 + \pi_0 }{(\pi_+ + \pi_- + \pi_i y_i^2)-(\pi_+ - \pi_- + \pi_i y_i)^2} \nonumber \\
&= \frac{1 - \pi_i(1 - y_i^2)}{1 - \pi_0 - \pi_i(1 - y_i^2) - (-\pi_- + \pi_+ + \pi_i y_i)^2} && (\pi \cdot \mathbf{1} = 1)\,. \label{equation28}
\end{align}
$\mathbf{1}$ is the all one vector.
Let $\kappa \in (0,1/2)$ be a constant to be fixed later, and let $\varepsilon = \kappa b^*$.

\paragraph*{Case 1} $0 \le y_i^2 < 1 - \varepsilon$. Continuing from \Cref{equation28},
\begin{align*}
\lambda_\infty &\ge \frac{1 - \pi_i(1 - y_i^2)}{1 - \pi_0 - \pi_i(1 - y_i^2)} && \text{neglecting second term in denominator}\\
&> \frac{1 - \pi_i \varepsilon}{1 - \pi_0 - \pi_i \varepsilon} && (1-y_i^2 >\varepsilon) \text{ and \Cref{fact:mediant}} \\
&= \frac{1}{1-\pi_0} + \frac{\pi_0 \pi_i \varepsilon}{(1-\pi_0)(1-\pi_0-\pi_i \varepsilon)} \\
& \geq \frac{1}{1-\pi_0} + \frac{\pi_0 \pi_i \varepsilon}{(1-\pi_0)(1-\pi_0)} \\
&= \beta + \beta(\beta-1) \pi_i \varepsilon && \frac{1}{1-\pi_0} = \beta \\
&= \beta + \beta(\beta-1) \pi_i \kappa b^* && \varepsilon = \kappa b^* \\
&\ge \beta + (\beta-1) \frac{\pi_i \kappa}{\sum_i p_i} && \text{Inequality \ref{min-imbalance}} \\
&\ge \beta +  \frac{\beta-1}{\beta} \cdot \frac{ \kappa}{(\sum_i p_i)^2} && \pi_i \ge \frac{1}{\beta \sum_i p_i}
\end{align*}

\paragraph*{Case 2} $1 \ge y_i^2 \ge 1-\varepsilon \Rightarrow |y_i| = y_i > 1 - \varepsilon$ we similarly have
\begin{align*}
\lambda_\infty &= \frac{1 - \pi_i(1 - y_i^2)}{1 - \pi_0 - \pi_i(1 - y_i^2) - (-\pi_- + \pi_+ + \pi_i y_i)^2} && \text{\Cref{equation28}} \\
&> \frac{1 + (-\pi_- + \pi_+ + \pi_i y_i)^2}{1 - \pi_0} && \text{\Cref{fact:mediant} and $\lambda_\infty > 1$}\\
&= \beta + \beta (-\pi_- + \pi_+ + \pi_i - (1-y_i) \pi_i)^2 \\
&\ge \beta + \beta |-\pi_- + \pi_+ + \pi_i| (|-\pi_- + \pi_+ + \pi_i| - 2|(1-y_i) \pi_i|) && (a-b)^2 \ge |a|(|a|-2|b|)\\
&\ge \beta + \beta b^*(b^* - 2\varepsilon) && *\\
&= \beta + \beta (b^*)^2 (1-2\kappa) && \varepsilon = \kappa b^*\\
&\ge \beta + \frac{1}{\beta} \cdot \frac{(1 - 2 \kappa)}{(\sum_i p_i)^2}\,. && \text{inequality  \ref{min-imbalance}} \\
\end{align*}

In inequality above, denoted by $*\,,$ we applied 
$|-\pi_- + \pi_+ + \pi_i| \ge b^*$ and further lower-bounded  $|-\pi_-+\pi_++\pi_i|-2|(1-y_i)\pi_i| > 0\,,$ considering $|-\pi_- + \pi_+ + \pi_i| \ge b^*$ and $2|(1-y_i)\pi_i| < 2\varepsilon \pi_i \leq 2\varepsilon < b^*$.

Substituting $\kappa = 1/3\,,$ we would have
$$
\lambda_\infty \ge \beta + \min \set{\frac{\beta - 1}{\beta} \cdot \frac{1}{3(\sum_i p_i)^2} , \frac{1}{\beta} \cdot \frac{1}{3 (\sum_i p_i)^2}} \,.
$$

\end{proof}

\subsection{$1+\varepsilon$ Approximation}\label{ssec:lambda-infty-star-approx}
We conclude complexity of computing $\lambda_\infty$ of the star by providing an efficient approximation algorithm.

\begin{theorem}\label{laminf-ptas}
For $\varepsilon \in (0, \min(0.1, \min_i \pi_i))$, there exists a $\text{poly}(n, \varepsilon^{-1})$ time algorithm computing a $(1 + \varepsilon)$-approximation for $\lambda_\infty$
of star graphs.
\end{theorem}
\begin{proof}
Searching for a (near) optimal $y$, we can bound the search space using the following assumptions. 
\begin{align*}
y_0 &= 0 && \text{uniform-shift invariance} \\
\max_k |y_k| &= 1 && \text{scale invariance} \\
\exists i \in [n-1], \ |y_k| &= 1, \forall{k \neq i}, y_i \in [-1,1] && \text{by \Cref{starchar}}
\end{align*} 

Without loss of generality we can assume the only (possibly) non-extreme leaf $i$ is known, suffering a multiplicative factor of $O(n)$ in the runtime. The decision to be made is on the value for $y_i \in [-1,1]$ along partitioning of the remaining leaves (i.e., $y_k \in \{ -1, +1\}$ for $0 < k \ne i$). 

Namely for $\pi_- = \sum_{k \neq i} \pi_k \cdot \mathds{1}[y_k = -1]$ and
$
\pi_+ = \sum_{k \neq i} \pi_k \cdot \mathds{1}[y_k = +1]\,,
$
the objective is to minimize
\begin{align}
\lambda_\infty = & \min_{\pi_-, \pi_+, y_i} \frac{\pi_- \cdot (-1)^2 + \pi_+ \cdot (1)^2 + \pi_i y_i^2 + \pi_0 \cdot (1)^2}{\pi_- \cdot (-1)^2 + \pi_+ \cdot (1)^2 + \pi_i y_i^2 - (\pi_- \cdot (-1) + \pi_+ \cdot 1 + \pi_i y_i)^2} \nonumber \\ 
= & \min_{\pi_-, \pi_+, y_i}  \frac{\pi_- + \pi_+ + \pi_i y_i^2 + \pi_0 }{\pi_- + \pi_++ \pi_i y_i^2 - (-\pi_- + \pi_+ + \pi_i y_i)^2}\,, \label{extremeformula}
\end{align}
where $\pi_0, \pi_i, \pi_- + \pi_+ = 1 - \pi_0 - \pi_i$ are known constants and the optimization is over a pair of valid $ -\pi_- + \pi_+ $ and $ y_i)\,.$

Without loss of generality, assume $\pi_- \ge \pi_+$ and $y_i \ge 0$ (else negating $y_i$ increases the denominator only), allowing to bound the variance at the optimal $y\,,$ from below.
\begin{align}
D \defeq  \var[\pi]{y} 
&= \sum_{j<k} \pi_j \pi_k (y_j-y_k)^2 \nonumber \\
&\ge \pi_0 (\pi_- + \pi_+) + \pi_i \pi_-  && y_0 = 0, y_i \ge 0 \nonumber \\
&\ge \frac{1}{2} (\pi_0 + \pi_i)(\pi_- + \pi_+) && \pi_- \ge \pi_+ \nonumber \\
&= \frac{1}{2} (\pi_0 + \pi_i)(1 - \pi_0 - \pi_i) && \sum_j \pi_j = 1 \nonumber \\
&\ge \frac{1}{2} 2 \varepsilon \cdot (1 - 2 \varepsilon) && \varepsilon < \min_k \pi_k \nonumber \\
&\ge \frac{1}{2} \varepsilon \cdot (1 - \varepsilon) && \varepsilon < 0.1 \label{boundvar}\,.
\end{align}
The following lemma paves the way for a dynamic programming solution. It shows minimizing $\lambda_\infty(\cdot)$ in a 
dense enough lattice over the search space for valid pairs of $-\pi_- + \pi_+$ and $y_i$ negligibly degrades the answer. 

\begin{lemma} \label{lemma12}
Given $d$ and $y_i'$ satisfying $|d - (-\pi_- + \pi_+)| < \varepsilon^2/100$ and $\max(0, y_i - \varepsilon^2/100) \leq y_i' \leq y_i$ we have
$$
\frac{\pi_- + \pi_+ + \pi_i y_i'^2 + \pi_0 }{\pi_- + \pi_++ \pi_i y_i'^2 - (d + \pi_i y_i')^2} \le \frac{\pi_- + \pi_+ + \pi_i y_i^2 + \pi_0 }{\pi_- + \pi_++ \pi_i y_i^2 - (-\pi_- + \pi_+ + \pi_i y_i)^2} (1+\varepsilon)\,.
$$
\end{lemma} 
\begin{proof}
Since $y_i' \leq y_i$ for each $i \in [n-1]$, we have
$${\pi_- + \pi_+ + \pi_i y_i'^2 + \pi_0 } \le {\pi_- + \pi_+ + \pi_i y_i^2 + \pi_0 }.
$$
Therefore, the numerator of $\lambda_\infty$ expression for $y'$ (LHS above) can be upper bounded by the numerator of the $\lambda_\infty$ expression for $y$ (in the RHS.) 
Now we lower bound the denominator of the $\lambda_\infty$ expression for $y'$.

We first observe that 
\begin{equation}
\label{eq:yi}
    y_i'^2 \geq y_i^2-2y_i\varepsilon^2/100,
\end{equation}
which can be verified considering two cases;
whether $y_i \ge \varepsilon^2/100$ or not. In the first case
$$y_i'^2 \ge \max(0,y_i-\varepsilon^2/100)^2 = (y_i-\varepsilon^2/100)^2 \ge y_i^2 - 2y_i\varepsilon^2/100\,.$$ In the other case, $0 \leq y_i'
\leq y_i \leq \varepsilon^2/100\,,$ hence we have $y_i'^2 \ge 0 \ge y_i(y_i-2\varepsilon^2/100) = y_i^2 - 2y_i\varepsilon^2/100$.

Next, we can show
\begin{equation}
\label{eq:12}    
    (d+\pi_i y_i')^2 \leq (|-\pi_-+\pi_++\pi_i y_i|+2\varepsilon^2/100)^2.
\end{equation}
This is due to the fact that given $|a-a'| \leq \alpha$ and $|b-b'| \leq \beta$ we have $(a+b)^2 \leq (|a'+b'|+\alpha+\beta)^2\,,$ which we applied for $a = d, a' = -\pi_-+\pi_+, \alpha = \beta = \varepsilon^2/100, b = \pi_i y_i', b' = \pi_i y_i$.

\begin{align}
D' & \defeq \pi_- + \pi_++ \pi_i y_i'^2 - (d + \pi_i y_i')^2 \nonumber \\
&\ge \pi_- + \pi_+ + \pi_i (y_i^2 - 2 y_i \varepsilon^2/100) - \left(|-\pi_- + \pi_+ + \pi_i y_i| + \frac{\varepsilon^2}{50} \right)^2 \label{ineq12}\\
&\ge \left(\pi_- + \pi_+ + \pi_i y_i^2 - \frac{\varepsilon^2}{50}\right) - (-\pi_- + \pi_+ + \pi_i y_i )^2 - \frac{\varepsilon^2}{50} \left(2 |-\pi_- + \pi_+ + \pi_i y_i|  + \frac{\varepsilon^2}{50} \right) \label{ineq13} \\ 
&\ge (\pi_- + \pi_++ \pi_i y_i^2 - (-\pi_- + \pi_+ + \pi_i y_i)^2) - \varepsilon^2/10 \label{ineq14}\\
&= D - \varepsilon^2/10\,. \nonumber
\end{align}
Note that inequality \eqref{ineq12} applied inequalities \eqref{eq:yi} and \eqref{eq:12}.
 Inequalities \ref{ineq13} and \eqref{ineq14} apply trivial bounds $\pi_i, y_i \in [0,1]$ and $\varepsilon \in (0,0.1)$.

Finally we have the desired
\begin{align*}
(1 + \varepsilon) D' &\ge (1 + \varepsilon)(D - \varepsilon^2/10) \\
&\ge D + \varepsilon D - \varepsilon^2/10 - \varepsilon^3/10 \\
&\ge D + \varepsilon (D - \frac{1}{10} \varepsilon (1 - \varepsilon)) \\
&\ge D + \varepsilon \cdot 0\, & \text{inequality \ref{boundvar}}.
\end{align*}
\end{proof}

Algorithm~\ref{starapproxalg} summarizes our dynamic programming solution. Enumerating over all choices for the special index $i\,, $ feasible balances of the remaining leaves are computed using dynamic programming over a dense quantization of the values, i.e., $dp[\cdot]$ is True for entries corresponding to a valid $\pi_-+\pi_+$.
\end{proof}

\begin{algorithm}
\caption{Approximating $\lambda_\infty$ of a star}
\begin{algorithmic} 
\REQUIRE Star graph $G = (V,E), \pi \in \Delta^n, \varepsilon \in (0, \min (0.1, \min_i \pi_i))$
\ENSURE $\lambda \in \lambda_\infty(G) \cdot [1, 1 + \varepsilon]$
\STATE $\lambda \leftarrow \infty$
\STATE $\pi' \leftarrow \lfloor \frac{100n}{\varepsilon^2} \pi \rfloor \frac{\varepsilon^2}{100n}$ 
\FOR {$i \in [n-1]$}
\STATE $dp \leftarrow \text{ all False for } -1:\frac{\varepsilon^2}{100n}:1$
\STATE $dp[0] \leftarrow \text{True}$
\FOR {$j \in [n-1] - \{i\}$}
\STATE {$dp[\cdot] = dp[\cdot] * \delta[\cdot-\pi'_i] \vee  dp[\cdot] * \delta[\cdot + \pi'_i] $} 
\ENDFOR
\FOR {$(d,y'_i) \text{ where } dp[d] = \text{ True and } y'_i \in 0:\frac{\varepsilon^2}{100}:1$}
\STATE {$\lambda \leftarrow \min(\lambda, \frac{\pi_- + \pi_+ + \pi_i y_i'^2 + \pi_0 }{\pi_- + \pi_++ \pi_i y_i'^2 - (d + \pi_i y_i')^2})$} \COMMENT{where $\pi_- + \pi_+ = 1 - \pi_0 - \pi_i$}
\ENDFOR
\ENDFOR
\RETURN $\lambda$
\end{algorithmic}
\label{starapproxalg}
\end{algorithm}

\subsection{Vertex Expansion}\label{ssec:vertex-expansion-tree}

We conclude this Section with a similar complexity result for vertex expansion.

\begin{proposition}\label{thm:vertex-exp-hard}
There exists a polynomial $p$ such that the following holds.
Given a star graph $G\,,$ and rational probability distribution $\pi$ over the vertices, computing $\phi^V(G,\pi)$ to $\text{p}(|G,\pi|)$ many bits accuracy is NP-hard (here $|G,\pi|$ denotes the length of the input string).
\end{proposition}

While complexity of computing $\lambda_\infty$ for unweighted trees remains an open problem, which we conjecture to be in P. We confirm that this is the case for vertex expansion.

\begin{theorem}\label{thm:vertex-exp-tree}
Given a tree $G$ and uniform distribution $\pi$ over the vertex set, the vertex expansion $\phi^V(G)$ can be computed in polynomial time.
\end{theorem}

\begin{proof}[Proof Sketch]
Root the tree from an arbitrary vertex. We propose a dynamic programming solution.

Let $dp[v][l_v][l_p][n_v][b_v]$ for some vertex $v \in V$, indicate whether one can partition the tree into some $A$ and $B = V \setminus A$ such that $v$ and its parent $p$ belong to $l_v, l_p \in \set{A, B}$ respectively, and $n_v$ of vertices in the subtree rooted at $v$ belong to $A$, and $b_v$ of vertices of this subtree are at the boundary of our partitioning (depending on inner/outer/symmetric definition for the vertex expansion). 

We can solve this $dp$ in post-order traversal due to a depth first search from the root of our tree. To compute values corresponding to every vertex $v$, $dp[v][\cdot][\cdot][\cdot][\cdot]$, we need to accumulate information corresponding to consistent $dp$ variables from children of $v$, to find all possible pairs of $n_v$ and $b_v$ using their previously computed $dp$ tables. This can be done by solving another dynamic programming task at every node $v$ to accumulate the information from its children.
\end{proof}

\section{Maximum Variance Embedding of a Tree}
\label{sec:MVE-tree}

In this Section we study the maximum variance embedding problem for trees. Trees can be optimally embedded into $k \ge 2$ dimensions by first solving the embedding problem in $n$ dimensions (that is equivalent to SDP \eqref{eq:primal-sdp} for $k = n$, hence convex) and then reducing the dimension to $2$ using the main algorithm of G\"oring, Helmberg, and Wappler \cite{GHW08}.
In \S\ref{ssec:MVE-tree-combinatorial} We extend their characterization of optimal MVE for trees and develop a linear time combinatorial algorithm, avoiding the expensive SDP step, to find an optimal MVE of the weighted tree in $\mathbb{R}^2$ (Theorem~\ref{2dtree}).
Further refuting the need to solve the SDP, in \S\ref{ssec:MVE-tree-BM} we prove a simple local search can also find the optimal maximum variance embedding in $\mathbb{R}^2$.

Little is known for MVE of trees into $\mathbb{R}^1$, equivalently the spread constant, where ordinary projection from a $2$ (or higher) dimensional solution gives no better than a $2$-approximation.
Let us further dismiss the SDP approach for approximating spread constant of a tree.
Consider a star with $3$ leaves, i.e., $S_4$, known as the claw graph, and let $\pi_0 = 0$ and $\pi_i = 1/3 : i \in [3]$. This example shows the \emph{integrality gap} for computing $\spcon$ using the SDP relaxation is $\ge 9/8$. Note that one can construct instances with arbitrary many vertices using the same structure.
In \S\ref{ssec:MVE-tree-combinatorial}, we provide a $1+\varepsilon$ approximation algorithm for $\spcon$ of weighted trees (\Cref{thm:fptas}).
Before presenting the algorithms, let us confirm that MVE of trees in $\mathbb{R}^1$ is an NP-hard problem.

\begin{theorem} 
\label{thm:sc-nph_theorem}
There exists a polynomial $p$ such that the following holds.
Given a star graph $G\,,$ and rational probability distribution $\pi$ over vertices, computing $\spcon(G,\pi)$ to $p(|G,\pi|)$ bits accuracy is NP-hard (here $|G,\pi|$ denotes the length of the input string).
\end{theorem}

To prove the above one can use the same star gadget as for $\lambda_\infty$. In fact, the analysis becomes further straightforward as we have a further restricted characterization of the optimal embedding.
The following lemma paves the way for a fully binary characterization of optimum embedding of a star. Similar to our analysis for $\lambda_\infty$, we can then reduce the {\sc Partition} problem to computing spread constant of a star. In particular we show 
$\spcon(S_n,\pi) \leq 1-\pi_0$, holding with equality if and only if the leaves can be equally partitioned.
For completeness, a proof of \Cref{thm:sc-nph_theorem} is included in \Cref{ssec:spread-constant-NPH}.

\begin{lemma}\label{full-stretch}
For any optimal embedding of a tree, all edges are fully stretched.
\end{lemma}

The lemma can be proved by noticing that if an edge is not fully stretched, moving (embedded) vertices of one component due to removal of the edge (in positive or negative direction) can increase the variance of the embedding.
Before discussing our algorithms, in \S\ref{ssec:variance} we elaborate on properties of the objective function, i.e., variance, that we are going to use more extensively from now on.

\subsection{On the Variance Function}\label{ssec:variance}

Variance of a random variable is its second moment about the mean, and is (half of) the expected squared difference of two i.i.d.\ samples of the random variable. For an embedding in $\mathbb{R}^1$, i.e., $y: V \rightarrow \mathbb{R}$, we can define variance as follows.
\[
\var[]{y} = \frac{1}{2}\expec{i,j \sim \pi}{(y_i-y_j)^2} = \expec{i \sim \pi}{y_i^2} - \expec{i \sim \pi}{y_i}^2 = \expec{i \sim \pi}{(y_i - \expec{j \sim \pi}{y_j})^2}
\]

Variance can be tracked independently in orthonormal bases as squared Euclidean distance tensorizes due to Pythagorean theorem. One can alternatively formulate the Euclidean variance of the embedding as
\begin{align}
\var[]{\yy} &\defeq \expec{i \sim \pi}{\|\yy_i - \expec{j \sim \pi}{\yy_j}\|^2} = \sum_{i < j} \|\yy_i - \yy_j\|^2 \pi_i \pi_j\\
&= \langle \text{Diag}(\pi)- \pi \pi^T , \YY \YY^T \rangle  = \sum_{k \in [d]} \var[\pi]{\YY \ee_k} = \var[]{\YY}\,.
\end{align}

Either interpreting the variance as the second moment about the mean, versus the (expected) squared distance of two independent samples can be more convenient/tractable in various circumstances. We further extend this characterization into a hybrid setup.
Before that, let us add to our notation. 
\begin{definition}
For an embedding $\yy: V \rightarrow \mathbb{R}^d$, the mean of the vectors (with respect to a distribution $\pi$) corresponding to a set of vertices $S \subseteq V$ is
$
\mu_{\yy}(S) \defeq \sum_{i \in S} \frac{\pi_i}{\pi_S} \yy_i
$,
where 
$
\pi_S = \sum_{i \in S} \pi_i
$.
\end{definition}

In our proofs, we are going to alter the embedding to increase the variance.
Considering a partition of vertices into disjoint subsets (that our alteration is isometric within each). The following lemma allows to only keep track of changes to the variance using distances between \emph{means} of subsets (rather than all pairwise distances).

\begin{lemma}\label{var-dif}
Consider a partition of vertices $V = \dot{\bigcup} S_i$ and two embeddings $\yy, \yy': V \rightarrow \mathbb{R}^d$ with identical Euclidean distances within every $S_i$, i.e.,
$$
\|\yy_u - \yy_v\| = \|\yy'_u - \yy'_v\| \quad \forall u,v \in S_i \quad \forall S_i.
$$
The difference in variance of $\yy$ and $\yy'$ can be tracked via means of $S_i$'s:
$$
\var[]{\yy} - \var[]{\yy'} = \sum_{i<j} \pi_{S_i} \pi_{S_j} \left( \left\|\mu_{\yy}(S_i)-\mu_{\yy}(S_j)\right\|^2 - \left\|\mu_{\yy'}(S_i) - \mu_{\yy'}(S_j)\right\|^2\right)\,.
$$
\end{lemma}

\begin{proof}
Variance can be written as
\[
\var[]{\yy} = \expec{v}{\|y_v - \mu_{\yy}\|^2} = \frac{1}{2} \expec{v,w}{\|\yy_v - \yy_w\|^2} = \sum_{v < w} \pi_v \pi_w \|\yy_v - \yy_w\|^2\,.
\]
Due to isometry within each $S_i$, while transforming $\yy$ to $\yy'$, we only need to track contributions of pairs of vertices from different $S_i$'s. Fix some $i < j$.
\begin{align}
\sum_{v \in S_i, w \in S_j} &\pi_v \pi_w \|\yy_v - \yy_w\|^2 \\ &= \sum_{v \in S_i, w \in S_j} \pi_v \pi_w \|\yy_v - \mu_{\yy}(S_i) + \mu_{\yy}(S_i) - \mu_{\yy}(S_j) + \mu_{\yy}(S_j) - {\yy}_w\|^2 \nonumber\\
& = \sum_{v \in S_i, w \in S_j} \pi_v \pi_w \left(\|{\yy}_v - \mu_{\yy}(S_i)\|^2 + \|\mu_{\yy}(S_i) - \mu_{\yy}(S_j)\|^2 + \|\mu_{\yy}(S_j) - {\yy}_w\|^2\right) \label{eq:terms1}\\
& \qquad\qquad + 2 \pi_v \pi_w ({\yy}_v - \mu_{\yy}(S_i))\cdot(\mu_{\yy}(S_i)  - \mu_{\yy}(S_j)) \label{eq:terms2}\\
& \qquad\qquad + 2 \pi_v \pi_w (\mu_{\yy}(S_i) - \mu_{\yy}(S_j))\cdot(\mu_{\yy}(S_j) - {\yy}_w) \label{eq:terms3}\\
& \qquad\qquad + 2 \pi_v \pi_w ({\yy}_v - \mu_{\yy}(S_i))\cdot(\mu_{\yy}(S_j) - {\yy}_w) \label{eq:terms4}
\end{align}
First note that the terms from \eqref{eq:terms2}, \eqref{eq:terms3}, \eqref{eq:terms4} have zero contribution to the summation. This is because (inner) product is linear, $\sum_{v \in S_i} \pi_v ({\yy}_v - \mu_{\yy}(S_i)) = 0$, and $\sum_{w \in S_j} \pi_v ({\yy}_w - \mu_{\yy}(S_j)) = 0$. Now, from \eqref{eq:terms1} the only terms that can differ in computing variance for ${\yy}$ and ${\yy}'$ are
\[
\sum_{v \in S_i, w \in S_j} \pi_v \pi_w \|\mu_{\yy}(S_i) - \mu_{\yy}(S_j)\|^2 = \pi_{S_i} \pi_{S_j} \|\mu_{\yy}(S_i) - \mu_{\yy}(S_j)\|^2
\]
that are exactly accounted for in the statement.
\end{proof}

The take-away from the above is that while applying transformations to subsets of points, besides keeping the solution feasible for corresponding maximum variance embedding, it suffices to not decrease the variance within each subset as well as the variance between corresponding means. A special case is when we are applying an infinitesimal isometric transformation to a subset of points that distances away their mean from the mean of other points, and equivalently from the original mean of all points, to increase the variance. In particular, consider a partition of vertices into two components due to removal of a single edge in the tree. Shifting one of the components in (at least) one direction (left or right) increases the variance of the embedding (considering the above Lemma, as the alternation will further separate the means). Hence we can feasibly increase the variance unless moving in target direction is violating the constraint of the removed edge. This gives a proof for Lemma~\ref{full-stretch}. Finally, we state a special case of Lemma~\ref{var-dif} as follows.

\begin{proposition}\label{lem:iso}
Applying an isometric transformation $f(\cdot)$ on a subset of the embedding $\yy(S), S \subsetneq V$, the variance of the new embedding $\zz: V \rightarrow \mathbb{R}^d$ can be written as
$
\var[]{\zz} = \var[]{\yy} + \pi_S \pi_{\bar{S}} \|\delta\|^2 + 2 \pi_S \pi_{\bar{S}} \langle \mu_{\yy}(S) - \mu_{\yy}(\bar{S}), \delta \rangle
$,
where $\delta = f(\mu_{\yy}(S))-\mu_{\yy}(S)$ is the shift on the mean of $S$.
\end{proposition}

\begin{proof}[Proof of Proposition \ref{lem:iso}]
Applying \cref{var-dif}, it suffices to prove the statement for the following special case. Let us shift the embedding by $\delta = f(\mu_{\yy}(S)-\mu_{\yy}(S)) \in \mathbb{R}^d$ for a subset of vertices $S \subseteq V$ that
changes the variance by
$
\var[]{\YY +  \mathbbm{1}_S \delta^T} - \var[]{\YY} = \pi_S \pi_{\bar{S}} \|\delta\|^2 + 2 \pi_S \pi_{\bar{S}} \langle \mu_{\yy}(S) - \mu_{\yy}(\bar{S}), \delta \rangle
$. The latter is due to $
\var[]{\YY + \DD} = \var[]{\YY} + \langle \text{Diag}(\pi)- \pi \pi^T , \DD \YY^T +  \YY \DD^T + \DD \DD^T \rangle
$
for $\DD = \delta^{\otimes n} \in \mathbb{R}^{n \times d}$.
\end{proof}

\subsection{Efficient Maximum Variance Embedding into $\mathbb{R}^1$ and $\mathbb{R}^2$}\label{ssec:MVE-tree-combinatorial}

\noindent\textbf{Approximation Scheme for MVE into $\mathbb{R}^1$.}
Alon, Boppana, and Spencer \cite{ABS98} showed there exists an {\em integral} optimal embedding $y \in \mathbb{R}^{V \times 1}$ for the spread constant of a general graph, with the following property. They showed there exists a set $U$ of vertices and an assignment of signs $s(C) \in \{\pm1\}$ corresponding to every connected component $C \subseteq V \setminus U\,,$ such that for every vertex $v$ in $C$, 
$$
y_v = d_G(U,v) \cdot s(C)
$$
and $y_u$ is zero for $u \in U$. 
This characterization facilitates an $O(2^{O(n)} \cdot \text{poly}(n))$ time algorithm to compute the spread constant.
While due to NP-hardness (\Cref{thm:sc-nph_theorem}) efficient computation of spread constant is unlikely, we show it can be well approximated, which we prove here.

\begin{theorem}\label{thm:fptas}
Given a tree $G$ and a distribution $\pi$ over vertices, the spread constant $\spcon(G, \pi)$ can be $1+\varepsilon$ approximated in time $\text{poly}(|G,\pi|, 1/\varepsilon)$, where $|G, \pi|$ denotes length of the input.
\end{theorem}

The key is to further enhance the characterization of optimal MVE for trees by the following Lemma. In words, we show that $U$ (in the characterization by \cite{ABS98}) can be further restricted as a singleton set, while this is far from true in general graphs (even for series-parallel graphs). 

\begin{lemma}\label{lem:tree_charac}
Let $y: V \rightarrow \mathbb{R}$ be an optimal embedding for spread constant of the tree, i.e., $\spcon(G,\pi)$. $y$ can be shifted by a constant such that all values are (in absolute value) equal to distances from a vertex. That is, $\exists v \in V$ such that $|y_w - y_v| = d_G(w,v) \quad \forall w \in V$.
\end{lemma}

\begin{proof}
Recall that all edges are stretched due to Lemma~\ref{full-stretch}.
We prove the desired by showing there is no simple path in $G$ where the corresponding $y$ sequence is \emph{tri-tonic}, i.e., increasing-decreasing-increasing (or decreasing-increasing-decreasing).

We prove this by contradiction. WLOG assume all edges are fully stretched, and we have a simple path $R = \set{R_1, \ldots}$ for which $y_R$ is increasing-decreasing-increasing. Consider the two critical positions of the path where monotonicity is broken, 
$$y_{R_i} < y_{R_{i+1}} > y_{R_{i+2}} \qquad \textrm{and} \qquad
y_{R_j} > y_{R_{j+1}} < y_{R_{j+2}}, \quad j > i\,.$$

WLOG assume $y_{R_{i+1}} = 0$ and $y_{R_{j+1}} = -d < 0$.
Let $A$ be the sub-tree including ${R_i}$ after removing ${R_{i+1}}$ from the tree, and let $B$ be the subtree including ${R_{j+2}}$ due to removal of ${R_{j+1}}$.
Let $C = V \setminus A \setminus B$ be the third subtree, together with $A$ and $B$ decomposing $V$.

We denote by $\mu_A$, $\mu_B$, $\mu_C$ as the means of the subtrees corresponding to $A$, $B$, and $C$, e.g.,
$$
\mu_A = \sum_{v \in A} \frac{\pi_v}{\pi_A} y_v.
$$
where $\pi_A, \pi_B, \pi_C$ are sum of corresponding probabilities.

First we show that $\mu_A = 0 - \alpha$ and $\mu_B = -d + \beta$ for some $\alpha, \beta > 0$, otherwise and in either case, the total variance can be increased by shrinking the edge between the corresponding subtree and $C$ or mirroring the subtree around the neighbor in $C$.
Lemma \ref{var-dif} shows all needed to secure an increase in the variance, is to depart the means.

Let us prove $\mu_A < 0$. First, see $\mu_A < \mu_{T \setminus A}$ as if otherwise we can increase $|\mu_A - \mu_{T \setminus A}|$ and the variance (applying Lemma \ref{var-dif}) by simply shifting $y_A$ to the right, i.e., shrinking the edge between $A$ and $T \setminus A$.
Now if $\mu_{T \setminus A} < 0$ we can increase $|\mu_{A} - \mu_{T \setminus A}|$ by mirroring $y_A$ around $0$ and applying Lemma \ref{var-dif}.

Similarly, to prevent increasing the variance by shifting $y_B$ to the left or mirroring it around $-d$,
$\mu_B > -d$ can be proved. So let $\mu_B = -d + \beta$.

Applying Lemma~\ref{var-dif}, let us write the increase (difference) in the variance due to mirroring both $y_A$ (around $0$) and $y_B$ (around $-d$) while we can reverse the effect of either by simply negating $\alpha$ or $\beta$ in the resulting formula.
\begin{align*}
\var[]{y'} - \var[]{y} &=
\pi_A \pi_B ((\alpha - (-d-\beta))^2 - ((-\alpha)-(-d+\beta))^2)\\
&+
\pi_A \pi_C ((\alpha-\mu_C)^2-(-\alpha-\mu_C)^2)+
\pi_B \pi_C ((-d-\beta - \mu_C)^2 - (-d+\beta - \mu_C)^2)\\
&=
\pi_A \pi_B 
(2d)(2\alpha + 2\beta) +
\pi_A \pi_C 
(-2\mu_c)(2 \alpha) +
\pi_B \pi_C
(-2d - 2\mu_C)(-2\beta)\\
&= \alpha
(4 \pi_A \pi_B d - 4 \pi_A \pi_C \mu_C) + \beta(4 \pi_A \pi_B d + 4 \pi_B \pi_C (d+\mu_C))\,.
\end{align*}
Hence, being free to decide signs on $\pm \alpha, \pm \beta$, we can increase the variance 
$$
|\alpha
(4 \pi_A \pi_B d - 4 \pi_A \pi_C \mu_C)| + |\beta(4 \pi_A \pi_B d + 4 \pi_B \pi_C (d+\mu_C))|\,,
$$
the sum of which will be strictly positive (former for $\mu_C \leq 0$ and latter for $\mu_C \ge 0$), so we have the desired contradiction and the original valuation was not optimal.
\end{proof}

The above characterization paves the way to prove Theorem~\ref{thm:fptas}. We can enumerate over all $O(n)$ possible cases for $U$. Let $U = \set{v}$. WLOG let $y_v = 0$ as the objective is shift invariant. We can write the maximization objective as
$$
\var[]{y} = (\expec{}{y^2} - \expec{}{y}^2)\,.
$$

$\expec{}{y^2}$ is only a function of $v$, i.e., invariant to $s(\cdot)$. All left is to minimize $(\expec{}{y})^2$ or equivalently $|\expec{}{y}|$. While all branches are fully stretched away from $y_v = 0$ and we only need to decide their positive/negative sign to minimize the sum (first moment) in absolute value.
Now the problem has become a knapsack which can be $1+\varepsilon$ approximated in fully polynomial time.

\paragraph{A Combinatorial Algorithm for MVE into $\mathbb{R}^2$.} As promised, we end this Section by providing a fast combinatorial algorithm to compute MVE of a (weighted) tree into $\mathbb{R}^2$. In this case we extend the following characterization of the optima of MVE in $\mathbb{R}^n$, stated as Proposition~\ref{tree-characterization} (that is interestingly similar to what happened for trees in $\mathbb{R}^1$, i.e., Lemma~\ref{lem:tree_charac}).

\begin{lemma}[\cite{GHW08}]\label{separator-shadow}
Let $y: V \rightarrow \mathbb{R}^n$ be a normalized (zero-mean) optimal solution to MVE in $\mathbb{R}^n$ and $S \subseteq V$ be a separator of the graph, removing which creates disconnected components $C_1, C_2 \subseteq V$. Then there exists $i \in \{1,2\}$ for which
$$
\text{conv}\{0,y_v\} \cap \text{conv}\{y_u : u \in S\} \ne \emptyset \quad \forall v \in C_i\,.
$$
\end{lemma}

Here, the convex hull of $S \subseteq \mathbb{R}^k$ is denoted by
$
\text{conv}(S) \defeq \{\alpha x + (1-\alpha) y : x,y \in S, 0 \leq \alpha \leq 1\}\,.
$

\begin{proposition}
\label{tree-characterization}
Given a tree $T = (V,E)\,,$ for any optimal MVE $y: V \rightarrow \mathbb{R}^2\,,$ either (i) the mean overlaps a single vertex $v\,,$ and 
for every other vertex $u \in V$ the graph distance between $v$ and $u$ matches that of the $\ell_2$ distance in the embedding, i.e.,
$$
d_T(u,v) = \|y_u-y_v\|_2
$$
or (ii) the mean belongs to a line segment corresponding to a single edge, i.e., $\expec{\pi}{y_v} \in \text{conv}\{y_u,y_v\}$ for some $(u,v) \in E\,,$ and the embedding spans only the line through $y_u$ and $y_v\,$ with all edges being stretched away from the mean.
\end{proposition}

\begin{proof}
{Without loss of generality, assume all edges are fully stretched and the mean is at the origin,
$
\expec{v \sim \pi}{y_v} = \mathbf{0}\,.
$
Applying Theorem \ref{separator-shadow}  for every vertex $v \in V\,,$ for which $y_v \neq \mathbf{0}\,,$ for all neighbors $u \in N(v)\,,$ except at most one, we have
$$
y_u = y_v + \frac{1}{\|y_v\|_2} y_v\,,
$$
as the shadow of separator vertex $v\,,$ with respect to the origin covers all and only the half line $\{(1+\alpha)y_v : \alpha > 0\}\,.$

If more than a single edge (the segment between two embedded vertices) or a single vertex overlaps the origin \Cref{separator-shadow} forces all  other edges connected to them to remain in their starting direction moving away from the origin. This would dismiss any other edges in the graph that potentially connect such rays that contradicts connectivity of the graph.

Finally note that if the line through a segment corresponding to an edge $(u,v)$ does not include the origin, their neighbors and the rest of the tree falls in diverging rays from the origin into $u$ and $v\,,$ and we will have a separating line (hyper-plane) between the origin and convex hull of $y\,,$ which contradicts $\mathbf{0}$ being the mean of $y\,.$

So one of the cases in the statement is valid.
}
\end{proof}

In \S\ref{ssec:MVE-tree-BM} we give another proof of the above lemma, using only local optimality conditions. We are now ready to present our combinatorial algorithm.

\begin{theorem}\label{2dtree} MVE of a (weighted) $n$-vertex tree into $\mathbb{R}^2$ can be computed in time $O(n)$.
\end{theorem}

\begin{proof}[Proof Sketch]
{
Assume  \Cref{tree-characterization} holds in Scenario (i) with the mean overlapping $y_v$ for some vertex $v$. Proposition \ref{tree-characterization} ensures the $\deg(v)$ subtrees (branches neighboring to $v$) are fully stretched away from $y_v$
and the mean of the embedding overlaps $y_v$. Let us check when this is feasible.
Contribution of every branch $b_i$ to total moment with respect to the mean, due to Proposition~\ref{tree-characterization}, has a magnitude 
$$
\sum_{u \in b_i} \pi_u d_T(v,u)
$$
which can be towards an arbitrary direction. Call this value $M_i$ corresponding to branch $b_i$. We need to decide whether there exists an arrangement of vectors of known magnitudes, i.e.,  $\set{M_i}$, that sum up to zero. Using the triangle inequality, observe that the necessary and sufficient condition for existence of the desired arrangement is the largest moment for these vectors to be no larger than the sum of all the rest, i.e., $2 \max_i \set{M_i} \leq \sum_i M_i$.

For scenario (ii) we have a similar situation where moments of the two branches, corresponding to endpoints of the edge containing the mean, need to cancel each other. Assume this edge is $e$, and compute total moment for each branch with respect to their endpoint, say $M_1$ and $M_2$. Displacing the origin (mean) away from an endpoint increases the moment of that branch (to the origin) at a rate of total probability-mass of that branch, say $\pi_1$ and $\pi_2$. One can find a feasible location along $e$ for the mean by solving a linear equation, $\alpha \pi_1 + M_1 = (1-\alpha) \pi_2 + M_2$. We have a feasible location for the mean if (and only if) the solution satisfies $\alpha \in (0,1)$.

Now, an $O(n^2)$ algorithm is straightforward, by trying the above $O(n)$ cases. For each case finding the moments and a feasible mean in can be done in $O(n)$, as well as computing the final answer, i.e., the second moment with respect to the mean if the case is feasible.

It is easy to see the above $O(n) + O(n)$ cases can be inspected (for feasibility of overlap with the mean) in a total time of $O(n)$, given pre-computed moments due to every branch neighboring every vertex. Before showing that this information can be pre-computed in $O(n)$, we denote another key observation that: one may find one and only one feasible mean among the above $O(n)$ cases. Thus we finish the algorithm by computing the second moment with respect to the feasible mean.

Performing two rounds of depth first search (DFS) from an arbitrary vertex on the tree, we can precompute moments and probability masses due to all neighboring branches out of every vertex. In the first DFS this information is computed in post-order and the second DFS provides the information corresponding to parent branch of every node to it, in pre-order.
}
\end{proof}

\subsection{A Benign Non-Convex Optimization Landscape}\label{ssec:MVE-tree-BM}
The first algorithm for maximum variance embedding 
\cite{WS04} solved the problem for $k = n$, i.e., the SDP over the Gram matrix of the embedding $\XX \in \mathbb{R}^{n \times n}$ where $\XX_{ij} = \yy_i \cdot \yy_j$. Then, it projected the solution down to the target dimension by applying a PCA. While the latter stage is still a heuristic (linear) transformation that can inevitably degrade the solution (see Figure~\ref{fig:star}) the former step faces a serious time and storage scalability challenge.
Our combinatorial algorithm from \S\ref{ssec:MVE-tree-combinatorial} addressed this by finding the $\mathbb{R}^2$ MVE in linear time.
The question that we address in the rest of this Section is whether this could be achieved by a ``simpler'' algorithm.

\begin{figure}[t]
    \centering
    \includegraphics[width=0.85\textwidth]{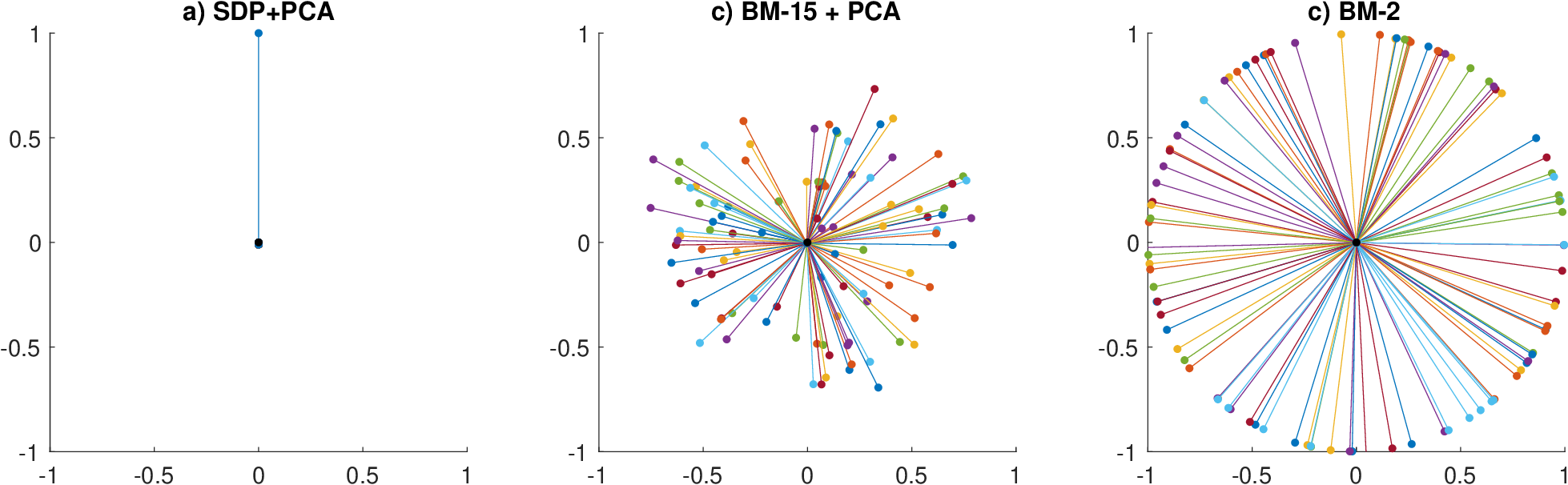}
    \caption{In this example $G = (V,E)$ is a star. We observe Burer-Monteiro approach not only provides a simple and memory-efficient algorithm, but also improves quality of the final answer.}
    \label{fig:star}
\end{figure}

A pioneering study by Burer and Monteiro 
\cite{BM03} suggested that when a SDP such as \Cref{eq:primal-sdp} has a solution of rank $k$, replace $\XX$ by $\YY \YY^T, \YY \in \mathbb{R}^{n \times d}$ in the formulation to simply impose the rank constraint, avoid positive-semidefinite matrices \& conic optimization, and use heuristic methods to find a \emph{locally optimal} solution for this constrained quadratic \emph{non-convex} optimization problem.
Despite significant performance of this heuristic in practice the question of ``how large \emph{must} we make $k$ so that the local optima are \emph{guaranteed} to map to global optima?'' has been partially addressed under various assumptions(see \cite{BNS16,CM19,BVB19}).

Applying a Grothendieck type inequality by \cite{MMMO17} to a new formulation of the problem (that we include \& analyze in Appendix~\ref{sec:MVE-tree-Grothendieck}) one can show that local optima are approximate global optima, stated as follows.
\begin{proposition}
Any locally optimal Maximum Variance Embedding of a tree in $k$ dimensions, is a $(1-\frac{1}{k-1})$-approximate solution w.r.t.\ the global optimum.
\end{proposition}
Nevertheless, existing non-convex optimization theory does not provide any guarantees for MVE in $\mathbb{R}^2$. We address this question by the following result, that we prove in the rest of this Section.

\begin{theorem}\label{thm:main1}
For $k \ge 2$ any local optima (for a Burer-Monteiro formulation) of MVE of a (weighted) tree is a global optimum. Moreover, the objective at such local (global) optima is equal to the SDP objective.
\end{theorem}

All we use to prove the above is first \textit{\&} second order necessary optimality conditions that we derive next. For more on the optimization theory that we use, one can refer to \cite{rockafellar1970convex, boyd2004convex}. For the rest of this Section, WLOG for trees and for the ease of notation, assume to only have upper bound inequality constraints for edges.
Define the Lagrangian function
\begin{align}
\mathcal{L}(\yy,w) \defeq  \var[i \sim \pi]{\yy_i} - \sum_{(i,j) \in E} w_{ij} (\|\yy_i - \yy_j\|^2 - \ell_{ij}^2)\,,
\end{align}
to rewrite the constrained optimization problem (MVE) as $$\max_{\yy} \min_{w: E \rightarrow \mathbb{R}^{\ge 0}} \mathcal{L}(\yy,w)\,.$$

\begin{lemma}[FONC]
\label{FONC}
At any local optima $\mathcal{L}(\yy, w)$,
First Order Necessary Conditions (FONC) a.k.a.\  KKT conditions \cite{K39, KT14} must hold, i.e., 
primal (Lipschitzness) and dual ($w \ge 0$) feasibility in addition to the following stationary and complementary slackness conditions for the Lagrangian.
\begin{equation}
    \begin{cases}
      \pi_i (\yy_i - \mu_{\yy}(V)) =  \displaystyle\sum\limits_{j: (i,j) \in E} w_{ij} (\yy_i - \yy_j) \quad \forall i \in V \\
      w_{ij} (\ell_{ij}^2 - \|\yy_i - \yy_j\|^2) = 0 \quad \forall (i,j) \in E \\
    \end{cases}       
\end{equation}
\end{lemma}

\begin{proof}
The stationary condition can be derived as follows while primal/dual feasibility and complementary slackness are easy to see.
\begin{align*}
\frac{\partial \mathcal{L}}{\partial \YY_{ik}} &= \sum_{j \ne i} 2 \pi_i \pi_j (\YY_{ik} - \YY_{jk}) -  \sum_{j: (i,j) \in E} 2 w_{ij} (\YY_{ik}-\YY_{jk})
\\
&= 2 \pi_i (\yy_{ik}-\mu_{\yy,ik}) - \sum_{j: (i,j) \in E} 2 w_{ij} (\YY_{ik}-\YY_{jk})
\end{align*}
\end{proof}

An intuitive interpretation of FONC is to think of the feasible solution (embedding) as placement of \emph{mass-points} $\set{\pi_i: i \in V}$ in the $d$-dimensional Euclidean space, $\yy: V \rightarrow \mathbb{R}^d$, while they are connected by \emph{strings}. By dual feasibility and complementary slackness, $w_{ij} \ge 0$ can be interpreted as a normalized \emph{stretch} of an edge $(i,j)$ that can be positive only if the corresponding string is \emph{fully} stretched. KKT conditions require this mass-string system to be  stationary while each stretched string is \emph{pulling} its endpoints proportional to the stretch and its displacement vector and each mass point is \emph{pushed away} from the mean by a force proportional to its mass and distance from $\mu_{\yy}$. This is visualized by Figure~\ref{fig:mass-rope}. 

\begin{figure}
  \begin{center}
\begin{tikzpicture}[line width=0.7pt, scale=0.7]
    \begin{scope}[shift={(-2.5,0)}]
    \draw [ultra thick, red!100] [->] (45:1) -- (45:2);
    
    \draw [ultra thick, red!100] [->] (135:1) -- (135:1.5);
    
    \draw [ultra thick, red!100] [->] (-45:1) -- (-45:1.5);
    
    \draw [ultra thick, red!100] [->] (10:1.5) -- (10:2.5);
    
    \draw [ultra thick, red!100] [->] (180+35:2) -- (180+35:3.5);
    
    
    \draw [ultra thick, black!50] [-] (45:1) -- (135:1);
    
    \draw [ultra thick, black!50] [-] (45:1) -- (-45:1);
    
    \draw [ultra thick, dotted, black!50] [-] (45:1) -- (10:1.5);
    
    \draw [ultra thick, dotted, black!50] [-] (180+35:2) -- ++ (0.3,0.6);
    \draw [ultra thick, dotted, black!50] [-] (180+35:2) -- ++ (0.6,0.3);
    \draw [ultra thick, dotted, black!50] [-] (135:1) -- ++ (-0.25,-0.4);
    
    
    \draw [ultra thick, black!100] [->] (45:1) -- ++ (-90:0.65);
    
    \draw [ultra thick, black!100] [->] (-45:1) -- ++ (90:0.65);

    \draw [ultra thick, black!100] [->] (45:1) -- ++ (180:0.65);
    
    \draw [ultra thick, black!100] [->] (135:1) -- ++ (0:0.65);


    \shade[ball color=blue](45:1) circle(4pt) node[above] {$\mathbf{y_1}$};
    \shade[ball color=blue](10:1.5) circle(3pt) node[below] {$\mathbf{y_2}$};
    \shade[ball color=blue](135:1) circle(3pt) node[left] {$\mathbf{y_3}$};
    \shade[ball color=blue](-45:1) circle(3pt) node[left] {$\mathbf{y_5}$};
    \shade[ball color=blue](180+35:2) circle(3pt) node[above] {$\mathbf{y_6}$};
    \filldraw[fill=red!20!white, draw=white](0:0) circle(8pt) node {$\mathbf{\mu_y}$};
    \node[text width = 2cm] at (1.1,-2.5) {(a)};
    \end{scope}

    \begin{scope}[shift={(4,0)}]

    \draw [ultra thick, red!100] [->] (1,0) -- ++ (1,0);
    \draw [ultra thick, red!100] [->] (-1,0) -- ++ (-1,0);
    \draw [ultra thick, red!100] [->] (0,0.2) -- ++ (0,0.1);
    \draw [ultra thick, red!100] [->] (0,-0.2) -- ++ (0,-0.1);
    
    \draw [ultra thick, black!50] [-] (0,0.1) -- (1,0);
    \draw [ultra thick, black!50] [-] (0,0.1) -- (-1,0);
    \draw [ultra thick, black!50] [-] (0,-0.1) -- (1,0);
    \draw [ultra thick, black!50] [-] (0,-0.1) -- (-1,0);
    
    \draw [ultra thick, black!100] [->] (0,0.1) -- ++ (0.5,-0.05);
    \draw [ultra thick, black!100] [->] (0,0.1) -- ++ (-0.5,-0.05);
    \draw [ultra thick, black!100] [->] (0,-0.1) -- ++ (0.5,0.05);
    \draw [ultra thick, black!100] [->] (0,-0.1) -- ++ (-0.5,0.05);
    \draw [ultra thick, black!100] [->] (-1,0) -- ++ (0.5,-0.05);
    \draw [ultra thick, black!100] [->] (-1,0) -- ++ (0.5,0.05);
    \draw [ultra thick, black!100] [->] (1,0) -- ++ (-0.5,-0.05);
    \draw [ultra thick, black!100] [->] (1,0) -- ++ (-0.5,0.05);

    \shade[ball color=blue](0,0.1) circle(4pt) node[above] {$\mathbf{y_1}$};
    \shade[ball color=blue](0,-0.1) circle(4pt) node[below] {$\mathbf{y_3}$};
    \shade[ball color=blue](1,0) circle(4pt) node[below] {$\mathbf{y_2}$};
    \shade[ball color=blue](-1,0) circle(4pt) node[below] {$\mathbf{y_4}$};   
    \node[text width = 2cm] at (1.1,-2.5) {(b)};
    \end{scope}
    
    \begin{scope}[shift={(10,0)}]

    \draw [ultra thick, red!100] [->] (0.707,0) -- ++ (0.707,0);
    \draw [ultra thick, red!100] [->] (-0.707,0) -- ++ (-0.707,0);
    \draw [ultra thick, red!100] [->] (0,0.707) -- ++ (0,0.707);
    \draw [ultra thick, red!100] [->] (0,-0.707) -- ++ (0,-0.707);
    
    \draw [ultra thick, black!50] [-] (0.707,0) -- (0,0.707);
    \draw [ultra thick, black!50] [-] (-0.707,0) -- (0,0.707);
    \draw [ultra thick, black!50] [-] (0.707,0) -- (0,-0.707);
    \draw [ultra thick, black!50] [-] (-0.707,0) -- (0,-0.707);

    \draw [ultra thick, black!100] [->] (0,0.707) -- ++ (0.353,-0.353);
    \draw [ultra thick, black!100] [->] (0,0.707) -- ++ (-0.353,-0.353);  
    \draw [ultra thick, black!100] [->] (0,-0.707) -- ++ (-0.353,0.353);
    \draw [ultra thick, black!100] [->] (0,-0.707) -- ++ (0.353,0.353);
    \draw [ultra thick, black!100] [->] (0.707,0) -- ++ (-0.353,-0.353);
    \draw [ultra thick, black!100] [->] (0.707,0) -- ++ (-0.353,-0.353);
    \draw [ultra thick, black!100] [->] (0.707,0) -- ++ (-0.353,0.353);
    \draw [ultra thick, black!100] [->] (-0.707,0) -- ++ (0.353,-0.353);    
    \draw [ultra thick, black!100] [->] (-0.707,0) -- ++ (0.353,0.353);

    \shade[ball color=blue](0.707,0) circle(4pt) node[below] {$\mathbf{y_2}$};
    \shade[ball color=blue](-0.707,0) circle(4pt) node[above] {$\mathbf{y_4}$};
    \shade[ball color=blue](0,0.707) circle(4pt) node[right] {$\mathbf{y_1}$};
    \shade[ball color=blue](0,-0.707) circle(4pt) node[left] {$\mathbf{y_3}$};
    \node[text width = 2cm] at (1.1,-2.5) {(c)};
    \end{scope}
 \end{tikzpicture}  \end{center}
 \caption{FONC/KKT conditions for MVE in $\mathbb{R}^2$. Every vertex $v$, at $y_v$, is exerted away from the mean by $\pi_v(\yy_v - \mu_y)$ drawn in red. Fully stretched ropes pull vertices at their endpoints. In (a) $y_1$ is in equilibrium. (b) and (c) show two equilibria (and optimal embeddings) for a 4-cycle.} 
 \label{fig:mass-rope}
\end{figure}

It is easy to see that FONC are not enough for (even local) optimality of a feasible solution. 
As an example consider a $Q_3$, i.e.,  a $3D$ cube, with edges of length $2$ and uniform $\pi$ over vertices. For MVE in $\mathbb{R}^2$, the optimal objective is $3$ as well.
One can verify an embedding $\yy_v = \pm \ee_1$, with mapping of neighboring vertices to opposite points, gives a \emph{stationary} solution (with valid stretches) with an objective value of $1$. Even an all-zero solution is stationary.
Nevertheless, given sufficient degrees of freedom, i.e., dimensions, such spurious yet stationary solutions will not be stable and a slight noise/perturbation leads to divergence (by KKT forces) from current solution. This instability can be embraced by Second Order Necessary Conditions (SONC).

Before formulating SONC, define $
E_{\yy} \defeq \set{(i,j) \in E : \|\yy_i - \yy_j\| = \ell_{ij}}
$ as the subset of edges for which, the constraints are active, i.e.,  fully stretched strings.
Let $\mathcal{T}_{\yy}$ be the corresponding tangent space, i.e., subspace of orthogonal perturbations to these active constraints, formulated as follows.

\begin{align}
\begin{split}
\mathcal{T}_{\yy} &\defeq \Bigl\{ \uu \in \mathbb{R}^{nd} : \\ & \quad \langle \nabla_{\yy} (\ell_{ij}^2-\|\yy_i - \yy_j\|^2) \, \uu \rangle = 0 ~ \forall (i,j) \in E_{\yy} \Bigr\}
\end{split}\\
\begin{split}
&= \Bigl\{ \uu = [\uu_1, \cdots, \uu_n]^T \in \mathbb{R}^{n \times d} : \\ & \quad \langle \yy_i - \yy_j \, \uu_i - \uu_j \rangle = 0 ~ \forall (i,j) \in E_{\yy} \Bigr\}
\end{split}
\end{align}

The following lemma formulates the SONC applied to our problem.
Adding to our notation, define
\begin{align}
\mathcal{M}( w) \defeq \sum_{(i,j) \in E} w_{ij} (e_i-e_j)(e_i-e_j)^T -  (\text{Diag}(\pi) - \pi \pi^T)\,,
\end{align}
using which one can write the main constraint of the dual SDP, \Cref{dualsdp}, as $\mathcal{M}(w) \succcurlyeq \mathbf{0}$. 
\begin{lemma}(SONC) \label{lem:SONC}
Any local optimizer $\mathcal{L}(\yy,w)$ (that is a regular point of the hypersurface defined by active constraints) must satisfy $\langle \mathcal{M}, \uu \uu^T \rangle \ge 0 \quad \forall \uu \in \mathcal{T}_{\yy}$ or equivalently
\begin{align}
\var[i \sim \pi]{\uu} &\leq \sum_{(i,j) \in E} w_{ij} \|\uu_i - \uu_j\|^2 \quad \forall \uu \in \mathcal{T}_{\yy}\,.
\end{align}
\end{lemma}

\begin{proof}
Second order necessary condition is
\begin{equation}
\langle \nabla^2_{\yy} \mathcal{L}(\yy,w) , \uu \uu^T \rangle \leq 0\,.
\end{equation}
Naming $\mathcal{L}_1(\yy,w) = \frac{1}{2} \var[i \sim \pi]{\yy_i}$ and $\mathcal{L}_2(\yy,w) = \frac{1}{2} \mathcal{L}(\yy,w) - \mathcal{L}_1(\yy,w)$ we can re-write the left hand side from the above as
\begin{align}
&\frac{1}{2} \langle \nabla^2_{\yy} \mathcal{L}(\yy,w) , \uu \uu^T \rangle \\&= \sum_{i,k} \frac{\partial^2 \mathcal{L}_1}{\partial \yy_{ik}^2} \uu_{ik}^2 + \sum_{i, j \ne i, k} \frac{\partial^2 \mathcal{L}_1}{\partial \yy_{ik} \partial \yy_{jk}} \uu_{ik} \uu_{jk} + \sum_{i, j \leftrightarrow i, k} \frac{\partial^2 \mathcal{L}_2}{\partial \yy_{ik}^2} \uu_{ik}^2 + \sum_{i, j \leftrightarrow i, k} \frac{\partial^2 \mathcal{L}_2}{\partial \yy_{ik} \partial \yy_{jk}} \uu_{ik} \uu_{jk}\\
&=
\sum_{i,k} \pi_i (1-\pi_i) \uu_{ik}^2 - \sum_{i, j \ne i, k} \pi_i \pi_j \uu_{ik} \uu_{jk} - \sum_{i, j \leftrightarrow i, k} w_{ij} \uu_{ik}^2 + \sum_{i, j \leftrightarrow i, k} w_{ij} \uu_{ik} \uu_{jk}\\
&=
\var[\pi]{\uu} - \sum_{(i,j) \in E} w_{ij} \|\uu_i - \uu_j\|^2\\
&= \langle - \mathcal{M}, \uu \uu^T \rangle
\end{align}
\end{proof}

In plain language, SONC requires that any perturbation that is relatively orthogonal to active (stretched) strings, to contain less variance w.r.t.\ mass-points than its second moment w.r.t.\ the stretches.

Notice that the above FONC and SONC conditions for $\mathcal{L}(\yy,w)$ are closely related to (strong) duality conditions for SDPs in \eqref{eq:primal-sdp} and \eqref{dualsdp}. Let $\XX_{ij} = \yy_i^T \yy_j$. It is easy to see that primal and dual feasibility of $\yy$ and $w$ for $\mathcal{L}$ ensures feasibility of $\XX$ for \eqref{eq:primal-sdp} and $w$ for \eqref{dualsdp}, except possibly the main constraint $\mathcal{M}(w) \succcurlyeq \mathbf{0}$. If we show that latter is satisfied, along with complementary slackness for the SDPs $\mathcal{M} \XX = \mathbf{0}$ (which is immediate due to FONC as $\mathcal{M} \XX = \mathcal{M} \YY \YY^T = \mathbf{0} \YY^T = \mathbf{0}$) strong duality for SDPs (which holds due to Slater's condition) ensures $\XX$ and $w$ are optimal solutions to SDPs with equal objectives. Same objective value is also achieved by $\yy,w$ for $\mathcal{L}$ (that is upper bounded by the SDP solution), hence $\yy,w$ is an optimal solution to the non-convex optimization problem. So we have proved the following Lemma.

\begin{lemma}\label{enough}
A feasible solution $\yy,w$ for $\mathcal{L}$ is at global optimum (with the same objective value as MVE in $\mathbb{R}^n$, i.e., \eqref{eq:primal-sdp} for $k = n$) if in addition to FONC (\Cref{FONC}),
$\mathcal{M}(w) \succcurlyeq \mathbf{0}$ is satisfied.
\end{lemma}

As an example of applying the above Lemma we can easily show the following.

\begin{proposition}\label{prop:ndim}
For MVE (of general graphs) in $\mathbb{R}^n$, any local optimum is a global optimum.
\end{proposition}

\begin{proof}
To apply \Cref{enough}, let us prove $\mathcal{M}(w) \succcurlyeq \mathbf{0}$ by contradiction. Assume $ \langle \mathcal{M}, \bb \bb^T \rangle < 0$ for some $\bb \in \mathbb{R}^n$ and WLOG $\mu_{\yy} = \mathbf{0}$. Consider unit vector $\dd \in \mathbb{R}^n$ orthogonal to $\Span{\yy{}}$. Let $\uu = \bb \dd^T \in \mathcal{T}_{\yy}$.  We have
$
\langle \mathcal{M}, \uu \uu^T \rangle = \langle \mathcal{M}, \bb \dd^T \dd \bb^T \rangle = \langle \mathcal{M}, \bb \bb^T \rangle < 0
$ which contradicts SONC for $\mathcal{L}(\yy,w)$.
\end{proof}

Before proving \Cref{thm:main1}, let us give yet another proof of \Cref{tree-characterization}, i.e., the characterization of local optima for MVE of trees in $\mathbb{R}^2$, using only FONC and SONC (and not \Cref{separator-shadow}).

\begin{proof}[Proof of lemma \ref{tree-characterization}]
We prove the characterization in two steps. \textbf{i:} Every edge is (fully stretched and) embedded over a line through the mean $\mu_{\yy}$. \textbf{ii:} All edges are stretched away from $\mu_{\yy}$. WLOG let $\mu_{\yy} = \mathbf{0}$.

\paragraph{i.}
Note that if some $w_{ij} = 0$ the SONC can be refuted simply by 
choosing an orthogonal vector $\bb$ to the (embedding of this) edge, $\bb \cdot (\yy_i - \yy_j) = 0$, and
perturbing vertices by $\uu_v = \bb \mathbbm{1} \cdot [d_G(v,i) < d_G(v,j)] \quad \forall v \in V$; whereas
\begin{align}
\langle \mathcal{M}, \uu \uu^T \rangle &= \sum_{e \in E} w_{e} \|\uu_{e_1} - \uu_{e_2}\|^2 - \var[\pi]{\uu}\\
&= 0 - \pi(C_u) \pi(C_v) \|\bb\|^2 < 0\,.
\end{align}
So $w_{ij} \ne 0 \quad \forall (i,j) \in E$ and complementary slackness shows every edge is fully stretched.
For the sake of contradiction, assume $(u,v)$ is embedded not along a line through the origin, i.e.,  $\dim{\Span{\yy_u, \yy_v}} = 2$.
Let $C_u, C_v$ be the corresponding components due to removal of the edge $(u,v)$ from $G$.

Archimedes showed the \emph{static equilibrium of a system}, corresponding to KKT-stationary condition for all of our vertices, requires not only sum of the forces on all of the points to be zero but also the net torque about any pivot to be zero. We use this fact to show $C_u$ (and $C_v$) is not in equilibrium, so not every point in the system satisfies KKT.
Let us compute the net torque on $C_u$, around the mean at $\mathbf{0}$. The arm for every $\pi_i \yy_i$ will be zero. Also pairs of torques due to $w_{ij}$ for $i,j \in C_u$ cancel each other. The only term left to sum is $\yy_u \times w_{uv}(\yy_u - \yy_v)$ which will be non-zero, giving the desired contradiction. This is due to $w_{uv} \ne 0$ and $\yy_v$ and $\yy_u$ being linearly independent.

\paragraph{ii.}
We claim all edges are stretched away from a single point, overlapping $\mu_{\yy}$, and rooting the tree from there will make all children of any vertex $s \in V$ to be embedded farther from the origin (root) than $s$. Otherwise, consider a vertex $s$ to have two neighbors $a$ and $b$ embedded at $\yy_a = \yy_s - \frac{\ell_{sa} \cdot \yy_s}{\|\yy_s\|}, \yy_b = \yy_s - \frac{\ell_{sb} \cdot \yy_s}{\|\yy_s\|}$ respectively. WLOG (due to symmetry) assume $\yy_s = y_s \ee_1$ for some $y_s > 0$. Let $C_a, C_b$ be the components due to removal of $s$ that contain $a, b$. Summing up KKT (stationary) conditions for vertices in $C_a$ results
$$
\mathbf{0} = w_{sa} (\yy_s-\yy_a)  + \sum_{i \in C_a} \pi_i \yy_i = w_{as} \ell_{sa} \ee_1 + \mu_{C_a} \pi(C_a)^{-1}\,.
$$
Due to the above we must have $\mu_{C_a} = - w_{as} \ell_{as} \pi(C_a) \ee_1 = -y_{C_a} \ee_1$. Note that $y_{C_a} > 0$ as otherwise leads to a contradiction with $y_{C_a} = w_{as} = 0$.
Similarly let $\mu_{C_b} = -y_{C_b} \ee_1$. We utilize SONC to show (infinitesimal) rotation of $C_a$ and $C_b$ around $s$ in opposite directions, i.e.,  towards $\pm \ee_2$, will increase the objective and refute local optimality. As we are not changing distances inside $C_a, C_b$, and $V \setminus C_a \setminus C_b$ it suffices to track the means, due to Proposition \ref{lem:iso}. Replace 
$C_a$ by a new vertex $1$ at $\yy_1 = -y_{C_a} \ee_1$ of mass $\pi_1 = \pi(C_a)$ and $C_b$ with $2$ at $\yy_2 = -y_{C_b} \ee_1 $ of mass $\pi_2 = \pi(C_b)$. 
Moreover, connect them to $s$ by edges of length $y_{C_a} + y_s$ and $y_{C_b} + y_s$, respectively.

Before writing SONC let us compute new $w$ as well. Using FONC, e.g.,  $-y_{C_a} \pi_1 + w_{1s} (y_{C_a} + y_s) = 0$, we will have
$
w_{1s} = \pi_1 \cdot \frac{ y_{C_b}}{y_{C_a} + y_s}$ and similarly $w_{2s} = \pi_2 \cdot \frac{y_{C_b}}{y_{C_b} + y_s}
$.
Finally we can see that SONC (\Cref{lem:SONC}) is refuted for $\uu_v = \ee_2 (\pi_2 \mathbbm{1}[v = 1] - \pi_1 \mathbbm{1}[v = 2]) \quad \forall v$, as
\begin{align*}
\var[]{\uu} &= \pi_1 \pi_2^2 + \pi_2 \pi_1^2 \\ &> \pi_1 \cdot \frac{ y_{C_a}}{y_{C_a} + y_s} \cdot \pi_2^2 + \pi_2 \cdot \frac{y_{C_b}}{y_{C_b} + y_s} \cdot \pi_1^2 \\ &= \sum_{ij} w_{ij} \|\uu_i - \uu_j\|^2\,.
\end{align*}
\end{proof}
We will now prove Theorem~\ref{thm:main1}, showing $\XX = \YY \YY^T,w$ at local optima of $\mathcal{L}$ to be the optimum with respect to the unconstrained SDP. Define a star $G'$ centered at a new vertex of mass $0$ and embed it at $\mu_{\yy}$ of $G$. For every $v \in V$ add a leaf to $G'$ with the same mass, embedded at the same point, and connected to the center by an edge of length $\|\yy_v-\mu_{\yy}\|$. One can observe $d_G(i,j) \leq d_{G'}(i,j) ~~ \forall i,j \in V$ is enforced due to the triangle inequality, hence MVE (value) of $G'$ upper bounds that of $G$, as the new maximization problem has weaker constraints by definition.
So it suffices to prove the claim for $G'$, i.e., a star.
The following lemma concludes the proof by showing MVE of $G'$ (in $\mathbb{R}^n$) is achieved by the variance of the same embedding in $\mathbb{R}^2$.

\begin{lemma}\label{lem:star}
Let $G$ be a star, $S_n$, with every leaf $i$ connected to the center $0$ by an edge of length $d_i$. Any solution $\yy: V \rightarrow \mathbb{R}^k$ that stretches all of the edges,  $\|\yy_0 - \yy_i\| = d_i \quad \forall i$, and is balanced around the center, $\mu_{\yy} = \yy_0$ is an optimal solution to MVE in $\mathbb{R}^n$.
\end{lemma}

\begin{proof}
If we show $\yy$ (zero-padded to $\mathbb{R}^n$) along with some $w$ satisfy FONC and SONC for $\mathcal{L}$ in $\mathbb{R}^n$, we can apply \Cref{prop:ndim} to ensure $\yy$ is global optimum to MVE in $\mathbb{R}^n$.

Let $w_{0i} = \pi_i \ge 0 \forall i$. It is easy to see FONC is satisfied. Considering any $\uu \in \mathcal{T}_{\yy}$ SONC holds as follows.
\begin{align*}
\var[\pi]{\uu} &= \sum_i \pi_i \| \uu_i - \expec{j}{\uu_j} \|^2\\
&\leq \sum_i \pi_i \| \uu_i - \uu_0 \|^2\\
 &= \sum_i w_i \| \uu_i - \uu_0 \|^2 = \sum_{ij} w_{ij} \|\uu_i - \uu_j\|^2
\end{align*}
\end{proof}

\section{Maximum Variance Embedding in Tree-width Dimensions}\label{sec:MVE-tree-width}
Generalizing Theorem~\ref{thm:sc-nph_theorem} from the previous Section, in \S\ref{ssec:MVE-tw-NPH} we show that computing/finding MVE in tree-width dimensions is NP-hard. In \S\ref{ssec:MVE-TW+2} we provide an algorithm that given any feasible solution to MVE of a graph along with a tree-decomposition of width $k$, provides a feasible solution to MVE in $\mathbb{R}^{k+2}$ with no less variance (objective value).

\noindent\textbf{Tree-width} is a common parameter in graph theory to measure the complexity of the graph structure. 
It has multiple equivalent definitions, from maximum size of a clique in a chordal completion of the graph, to size of the largest node in an optimal tree-decomposition, see \cite{diestel2005graph,robertson1986graph}. Interestingly, in a first discovery, tree-width was named the \emph{dimension} of the graph  \cite{bertele1973non}.
Graphs of small tree-width appear in many artificial structures by design, e.g., the classical routine to compute resistance of a series-parallel electric circuit applies to series-parallel graphs, i.e., graphs of tree-width at most $2$. Many natural data-sets can pose low tree-width structures, e.g., social networks \cite{adcock2016tree}. From theoretical perspective, bounded-tree-width graphs are specially interesting for making many hard problems easier to solve.

\subsection{MVE in Tree-Didth Dimensions is NP-hard}\label{ssec:MVE-tw-NPH}

Here we prove that MVE in tree-width dimensions is NP-hard, in the special case of having only inequality (upper bound) constraints (that makes the problem feasible for embedding into any Euclidean space, while equality constraints may prevent feasibility in less than tree-width dimensions). In the rest of this subsection, we prove the following result.

\begin{theorem}\label{thm:hard}
Given a tree-decomposition of width $k$ of a weighted graph, finding its maximum variance embedding (with only upper bound inequality constraints) into $\mathbb{R}^k$ is NP-hard.
\end{theorem}

Let us sketch the ideas behind our proof. The reduction will be from the same decision problem {\sc Partition} as for Theorem~\ref{thm:sc-nph_theorem}, given by $N$ integers $P_1, \cdots, P_N$, and querying whether the numbers can be partitioned into two sets of equal sum.
Here, the gadget is a $k$-star, that we define next.

A $k$-tree is a graph due to repeatedly adding new vertices to an initial $k$-clique, and connecting the new vertex to exactly $k$ existing ones.
(For more information on maximal graphs of tree-width $k$, i.e., $k$-trees, see \cite{patil1986structure}.)
Similarly, we call the graph a $k$-star when the $k$ neighbors of all new vertices must be the initial $k$ vertices. 
A $k$-star for $k = 3$ is depicted by Figure~\ref{fig:pyramids}-b, at (proximity) of a candidate optimal embedding into $\mathbb{R}^3$ that we are going to characterize next. First, note that the tree-width of the gadget is $\leq k$. 
\begin{fact}
Tree-width of a $k$-star is (at most) $k$.
\end{fact}

\begin{figure}
\centering
    \includegraphics[width=0.45\textwidth]{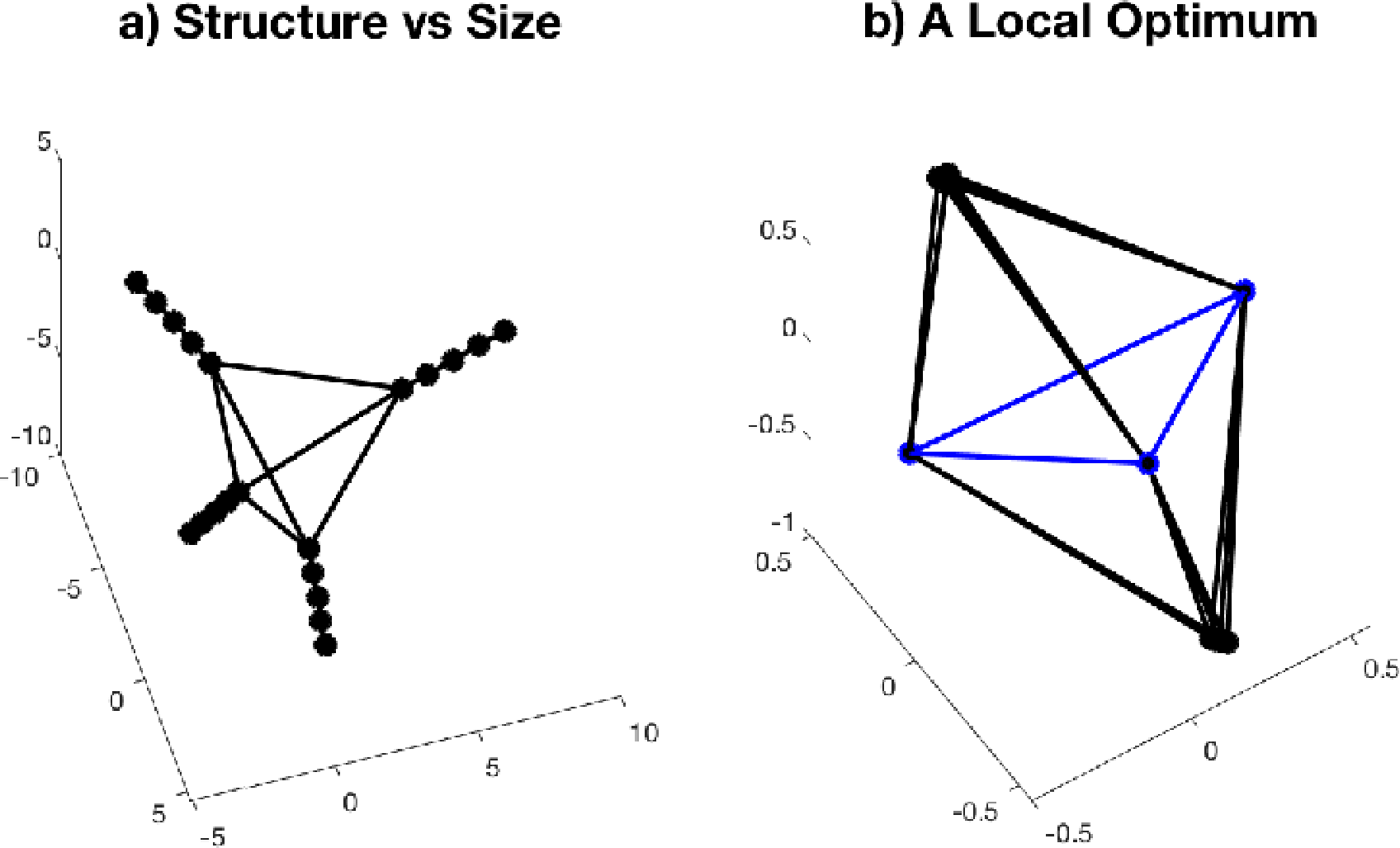}
    \caption{\textbf{a.} A large graph with a benign MVE optimization landscape in $\mathbb{R}^3$. \textbf{b.}
    In the neighborhood of a spurious local optima of maximum variance embedding of a
    a $3$-star in $\mathbb{R}^3$, which is also NP-hard in general.}
    \label{fig:pyramids}
\end{figure}

We engineer the $k$-star, such that any (locally) optimal embedding would put the central $k$ vertices away from each other, into the corners of a simplex (scaled $\Delta^{k-1}$) and pushes the leaves away from the simplex, along the remaining orthogonal dimension, in positive or negative direction.

Consequently, an equal partition of leaf weights, as in Figure~\ref{fig:pyramids}-b, will lead to a global upper bound of the maximum variance embedding problem for the $k$-star, that enables us to decide {\sc Partition} using a polynomially deep bit of the maximum variance embedding objective. Here is the proof.

\begin{proof}[Proof of Theorem~\ref{thm:hard}]

We prove the hardness statement for the special case when the edges have only upper bound constraints. The feasible region for this optimization problem has genus $0$, for which the gadget presents exponentially spurious local optima.
Same hardness result applies to equality constrained problem, and is easier to prove. Here, NP-hardness is shown with respect to vertex weights, $\pi$, assuming unit distance of all edges, while the same technique allows one to show similar NP-hardness with respect to $d$, with uniform $\pi$.

\textbf{The gadget.}
As promised we have a $k$-star. Let $u_1, \cdots, u_k$ be the initial $k$-clique, with uniform vertex weights $\pi(u_j) = {(1-\alpha)}/{k}$, where $\alpha$ will be set later.
Given an instance of the integer {\sc Partition} problem, $\{P_1, \cdots, P_N\}$, add a vertex $v_i$ corresponding to each $P_i$ and connect it to all $u_i$'s. Let $\pi(v_i) = \alpha \cdot \frac{P_i}{\sum_i P_i}\,$ where $\alpha \in (0,1)$. The gadget can be described as
$$
V = \big\{v_i : i \in [N]\big\} \cup \big\{u_j : j \in [k]\big\}
$$
and
$$
E = \bigg\{(v_i,u_j): \forall (i,j) \in [N] \times [d] \bigg\} \cup \bigg\{(u_i,u_j): \forall (i,j) \in {[d] \choose 2} \bigg\}\,.
$$

A key step is to show the optimal maximum variance embedding of the $k$-star into $\mathbb{R}^k$ fully stretches all of the edges. Note that this is not even feasible in general, e.g., embedding of a $3$-cycle in $\mathbb{R}^1$. Even if feasible, in general this is not necessarily optimal, e.g., add a chord to a $4$-cycle, put all the mass into the non-adjacent pair, and embed the graph in tree-width dimensions, i.e. $\mathbb{R}^2$. Let us characterize the optima.

\begin{lemma}
\label{lem:stretch-NPH}
Any optimal $\yy$ for maximum variance embedding of our $k$-star gadget into $\mathbb{R}^k$ would fully stretch all of the edges, i.e.,
$$
\|\yy_{v_i} - \yy_{u_j}\|_2 = 1 \quad \forall (i,j) \in [N] \times [d]
$$
and
$$
\|\yy_{u_i} - \yy_{u_j}\|_2 = 1 \quad  \forall (i,j) \in {[d] \choose 2}\,.
$$
\end{lemma}

\begin{proof}

Proving by contradiction, we have two cases. Let the smallest edge of the simplex be $\|\yy_{u_k} - \yy_{u_t}\| < 1$ for some $k,t \in [k]$.
First, we show $\|\yy_{u_k} - \yy_{u_t}\| > 1-\varepsilon_k$. If not, the variance for a new embedding $\tilde{\yy}$ that assigns $\tilde{\yy}_{u_i} = \frac{1}{\sqrt{2}} e_i \quad \forall i \in [k]$ (and maps $v_j$'s to any point inside the convex hull of the simplex) will be strictly larger.
\begin{align}
\var[]{\tilde{\yy}} - \var{\yy}
&\ge (\frac{1-\alpha}{k})(\frac{1-\alpha}{k})(1^2-(1-\varepsilon_k)^2) + \alpha (1-\alpha) (0^2-1^2) + \alpha (0^2 - 2^2)\\
&> \varepsilon_k \cdot (1-\alpha)^2/k^2 - 5 \alpha\\
&> 0 \qquad \text{holds for $\alpha / (1-\alpha)^2 < \varepsilon_k / 5k^2$, so  sufficient if $\alpha < \varepsilon_k/10k^2$}\,.
\end{align}
Here we underestimated the increase in the variance due to increasing the distance between other $u_i$ and $u_j$ and neglected any gains from covariance between $u$'s and $v$'s and the variance of $v$'s.

Remains the case where $\|\yy_{u_k}-\yy_{u_t}\| \in (1-\varepsilon_k, 1)$. We show $\yy$ is not even locally optimal. For infinitesimal $\varepsilon>0$, there exists an embedding $\tilde{\yy}:V \rightarrow \mathbb{R}^k$ where
$$
\|\tilde{\yy}_{u_k}-\tilde{\yy}_{u_t}\| = \|\yy_{u_k}-\yy_{u_t}\| + \varepsilon\,
$$
while other edges of the simplex preserve length and
$$
\|\tilde{\yy}_{u_i} - \yy_{u_i}\| \leq \varepsilon \quad \forall i \in [k]\,.
$$
Moreover, $\tilde{\yy}_{v_j}$ can be chosen no farther than $k \varepsilon$ from $\yy_{v_j}$ to preserve Lipschitzness for $v_j$'s. The differential change on variance will be at least
\begin{align*}
\lim_{\varepsilon \rightarrow 0} \frac{\var[]{\tilde{\yy}} - \var[]{\yy}}{\varepsilon} &=  (\frac{1-\alpha}{k})^2 \cdot 2(1-\varepsilon_k) - \alpha (1-\alpha) \cdot 2 (k+1) - \alpha 2 (2) 2k\\
&> (1-\alpha)^2/k^2 - 12 \alpha k && \text{$\varepsilon_k < 1/2$}\,.
\end{align*}
The above is positive for $\alpha < (1-\alpha)^2/12k$ for which $\alpha < 1/20k$ is sufficient. In this case a small increase in $\|\yy_{u_t} - \yy_{u_k}\|$ leads to an increase in variance hence we conclude $$\|\yy_{u_i} - \yy_{u_j}\| = 1 \qquad \forall i,j \in [k]\,.$$

Now we show $\|\yy_{v_i} - \yy_{u_j}\|_2 = 1$ for all $i \in [N], j \in [k]$. Proving by contradiction, if less than $k$ neighbors of a $\yy_{v_i}$ are fully stretched, intersection of tangent hyperplanes to hyperspheres corresponding to stretched neighbors will have dimension at least one, so there will be a line (first order approximate of a curve) along which we can move $\yy_{v_i}$ in both directions. Moving in at least one direction will have positive inner product with $\yy_{v_i} - \expec{}{\yy}$ so will increase the variance.
\end{proof}

Applying the above lemma, we know
$\yy_{u_1}, \cdots, \yy_{u_k}$ are $k$ equidistant points in $\mathbb{R}^k$ with respect to $\ell_2$ norm, so they are vertices of a regular simplex $\frac{1}{\sqrt{2}} \cdot \Delta^{k-1}$. Without loss of generality we can assume $\yy_{u_i} = \frac{1}{\sqrt{2}} e_i$ where $\{e_1, \cdots, e_k\}$ is an orthonormal basis for $\mathbb{R}^k$.

Forcing $\yy_{v_i}$ to be at equal distances from all vertices of $\frac{1}{\sqrt{2}} \cdot \Delta^{k-1}$ is equivalent to $\yy_{v_i} \in \Span{e}\,$ i.e., $\yy_{v_i} = \beta_i \sum_{j \in [k]} e_j$. Enforcing $\|\yy_{v_i} - \yy_{u_1}\| = 1$ results
\begin{align}
    1 &= \|\yy_{v_i} - \yy_{u_1}\|_2^2\\
    &= (\beta_i-\frac{1}{\sqrt{2}})^2 + (k-1) \beta_i^2 \\
    &= d \beta_i^2 - \sqrt{2} \beta_i + \frac{1}{2} \,.
\end{align}

Solving the quadratic eqality $k \beta_i^2 - \sqrt{2} \beta_i - \frac{1}{2} = 0$ for $k \ne 0$ gives
$$
\beta_i \in \setof{\frac{1 \pm \sqrt{k+1}}{\sqrt{2} k}}\,.
$$

Denoting $\pi_- = \sum_{i \in [N]} \pi(v_i) \cdot \mathds{1}[\beta_i = \frac{1 - \sqrt{k+1}}{\sqrt{2} k}]$ and $\pi_+ = \sum_{i \in [N]} \pi(v_i) \cdot \mathds{1}[\beta_i = \frac{1 + \sqrt{k+1}}{\sqrt{2} k}]$
 we can write the variance of the embedding as
\begin{align}
\var[\pi]{\yy} &= \frac{k \cdot (k-1)}{2} (\frac{1-\alpha}{k})^2 + \pi_- \pi_+ (\frac{2\sqrt{k+1}}{\sqrt{2}k})^2 + \alpha(1-\alpha) (1)^2 \\
&= \frac{(k-1)(1-\alpha)^2}{2k} + 2\pi_- \pi_+ \frac{k+1}{k} + \alpha(1-\alpha) \\
&= \frac{(k-1)(1-\alpha)^2 + 2k \alpha (1-\alpha)}{2k} + 2\pi_- \pi_+ \frac{k+1}{k}\,.
\end{align}

Considering $\pi_- + \pi_+ = \alpha$ we can see the maximum of $\var[\pi]{\yy}$ is achieved for (and only for) $\pi_- = \pi_+ = \alpha/2$ as
\begin{align}
2\pi_- \pi_+ \frac{k+1}{k} & \leq  \frac{(\pi_- + \pi_+)^2}{2} \cdot \frac{k+1}{k} && \text{AM-GM inequality} \\
&= \alpha^2 \frac{k+1}{2k}\,.
\end{align}

The maximum variance of $\frac{(k-1)(1-\alpha)^2 + 2k \alpha (1-\alpha)}{2k} + \alpha^2 \frac{k+1}{2k}$ is the answer if and only if the Partition is feasible. Otherwise $$|\pi_- - \pi_+| \ge \alpha \cdot \frac{1}{\sum_i{P_i}}\,.$$
In this case we have a decrease in variance, compared to the aforementioned maximum, by at least
\begin{align}
\frac{k+1}{2k} \cdot ( (\pi_- + \pi_+)^2 - 4\pi_- \pi_+ ) &\ge
\frac{k+1}{2k} \cdot (\pi_- - \pi_+)^2  \\
&\ge \frac{k+1}{2k} \cdot (\frac{\alpha}{\sum_i P_i})^2\,.
\end{align}

This will affect a polynomially deep bit of the answer and enables a reduction from the partition problem to decision version of maximum variance embedding as desired. 
\end{proof}

\subsection{MVE in Tree-Width$+2$ Dimensions}\label{ssec:MVE-TW+2}

We end this Section we present a dimension reduction algorithm for MVE with equality and/or inequality constraints.
The input is assumed to be any feasible solution $\yy: V \rightarrow \mathbb{R}^{\ell}$ to MVE (potentially the optimum in $\mathbb{R}^n$ found by solving the SDP) along with a tree decomposition $\mathcal{T}$ of width $k$ for the graph. 
The output will be a feasible solution $\zz: V \rightarrow \mathbb{R}^{k+2}$ with no less variance. For feasibility, we simply preserve distances between any pair of vertices that share a node in $\mathcal{T}$. 

Adding more to the notation, for a set of vectors $\xx$, $\dim{\xx}$ denotes the dimension of their affine hull, $\Aff{\xx}$. Our recursive algorithm is based on the following lemma.

\begin{lemma}\label{lem:base1}
Consider a subset of vertices $S \subseteq V$ of size $|S| \leq k$. If for every component $A \subseteq  V \setminus S$ of $G \setminus S$, $\dim{\yy_{A \cup S}} \leq k+1$, $\zz$ can be found as desired.
\end{lemma}

\begin{proof}
Without loss of generality we can assume $\zero \in \Aff{\yy_S} = \Span{\yy_S}$, due to shift invariance of the objective function (and constraints).
Let $L_S = \Span{\yy_S}$ be a linear subspace of dimension $r \leq |S|-1 \leq k-1$, with an orthonormal basis $\uu_1, \cdots, \uu_r$.

\begin{lemma}\label{lem:transform}
Given linear subspaces $L$ and $L'$ in the Euclidean space where $\text{rank}(L) \leq \text{rank}(L')$, there exist isometric transformation $U^L_{L'}$ that injects $L$ into $L'$ and is identity over $L \cap L'$.
\end{lemma}

\begin{proof}
Let $r = \rank(L), r' = \rank(L')$, and $r'' = \rank(L \cap L')$, and consider an orthonormal basis $\{\uu_1, \cdots, \uu_{r''}\}$ for $L \cap L'$. Expand this into an orthonormal basis $\{\uu_1, \cdots, \uu_{r''}, \uu_{r''+1}, \cdots, \uu_{r}\}$ for $L$, and similarly $\{\uu_1, \cdots, \uu_{r''}, \uu_{r''+1}, \cdots, \uu_{r'}\}$ for $L'$. Let $U_{L'}^{L}$ be a unitary transformation (i.e.,\ a change of basis) in the input space that maps $\uu_i \mapsto \uu_i ~ \forall i \in \{1,\cdots,r''\}$, $\uu_i \mapsto \uu_i' ~ \forall i \in \{r+1, \cdots, r\}$; which is indeed feasible because $r \leq r'$.
\end{proof}

Using Lemma~\ref{lem:transform} one may observe that every $\yy_{A \cup S}$ can be injected into a single linear superspace of $L_S$ of dimension $k+1$, preserving distances within each $\yy_{A \cup S}$.
Let us show this in detail.
For $v \in S$, let $\zz_v = (\yy_v \cdot \uu_1, \cdots, \yy_v \cdot \yy_v \cdot \uu_r, 0, 0, \cdots)$.
For each component $A \subseteq V \setminus S$, expand the basis to an orthonormal basis for $\Span{\yy_{S \cup A}}$, say $\uu_1^A = \uu_1, \cdots, \uu_r^A = \uu_r, \uu^A_{r+1}, \cdots, \uu^A_{k+1}$. For $v \in A \cup S$ let $\zz_v = (\yy_v \cdot \uu_1^A, \cdots \yy_v \cdot \uu_{k+1}^A)$.
Note that this is consistent with what we have for $\zz_S$.

As $\zz_{S \cup A}$ is a change of basis for $\yy_{S \cup A}$, pairwise distances within every $\yy_{S \cup A}$ are preserved. 
Considering Lemma~\ref{var-dif}, we only need to keep track of the variance due to pairwise distances between the new means, i.e., replace every $\yy_A$ with $\mu_{\yy}(A)$, that will lie on the same $\Span{\yy_A}$ and look at the objective for the new instance. We need to ensure the variance due to the transformed $\mu_{\zz}(A)$'s is no less.

Denote the projection of vectors in $\xx$ into linear a space $L$ (or its orthogonal complement) by $p_L(\xx)$ (or $p_L^{\perp}(\xx)$).
$p_{L_S}(\mu_{\yy}(A))$ is transformed into $p_{\mathbb{R}^k}(\mu_{\zz}(A))$, and $p_{L_S}^{\perp}(\mu_{\yy}(A))$ is transformed into $p_{\mathbb{R}^k}^{\perp}(\mu_{\zz}(A))$. Tensorizing the variances, we have $\var[]{\yy} = \var[]{p_{L_S}(\yy)} + \var[]{p^{\perp}_{L_S}(\yy)} $, and $\var[]{\zz} = \var[]{p_{\mathbb{R}^r}(\zz)} + \var[]{p_{\mathbb{R}^r}^{\perp}(\zz)}$.
Considering that the same unitary transformation is applied to $L_S$, the first term in variances is equal, i.e., $\var[]{p_{L_S}(\mu_{\yy}(A))} = \var[]{p_{\mathbb{R}^r}(\mu_{\zz}(A))}$.
All left is to ensure 
$$\var[]{p_{L_S}^{\perp}(\mu_{\yy}(A))} \leq \var[]{p_{\mathbb{R}^r}^{\perp}(\mu_{\zz}(A))}\,.$$

Left hand side of the above inequality is the variance due to having masses $\mu_A$ at 
$p_{L_S}^{\perp}(\mu_{\yy}(A))$. Indeed this variance is upper bounded by the maximum variance that one can achieve by embedding these probability masses $\mu_A$ in an arbitrary dimensional space, constrained to distances $\|p_{L_S}^{\perp}(\mu_{\yy}(A))\|$ from the origin. This is a star problem! It is easy to show that embedding all leaves of a star into at most $3$ directions can maximize the variance, i.e., a global optima spanning at most $2$ dimensions. Consider these $3$ directions in the $r+1$'th and $r+2$'th dimension of the $\zz$ space. Indeed they exist as $r+2 \leq k+1$. If every $\|p_{L_S}^{\perp}(\mu_{\yy}(A))\|$ was transformed into this target direction $\var[]{p_{\mathbb{R}^r}^{\perp}(\mu_{\zz}(A))}$ would have achieved the upper bound of what $\var[]{p_{L_S}^{\perp}(\mu_{\yy}(A))}$ could be. We can make this happen by applying another unitary transformation for every $A$, that is identity over over the first $r$ coordinates and maps existing $p_{L_S}^{\perp}(\mu_{\zz}(A)) \in (\mathbb{R}^k)^{\perp}$ into the desired direction in $(\mathbb{R}^r)^{\perp}$.
\end{proof}

\begin{example}
Let $S$ be the $k$ central vertices of the $k$ star, to which other vertices are connected to. Applying Lemma~\ref{lem:base1} shows that $k$-star has an optimal solution of dimension $\leq k+1$.
\end{example}

\paragraph{The recursion.}
The algorithm reduces the problem into smaller instances (of number of graph vertices or tree-decomposition nodes) and uses Lemma~\ref{lem:base1} to merge the solutions.

Without loss of generality $\mathcal{T}$ is minimal, i.e., no node is the subset of a neighbor (otherwise it can be removed).

If the diameter of $\mathcal{T}$ is $0$, i.e., having a single node, the number of vertices is at most $k+1$, so is $\dim{\yy} \leq k$, and we are done. (Let $\zz$ be representation of $\yy$ in an orthonormal basis of $\Aff{\yy}$.)

If the diameter is $1$, we have two nodes, say $A$ and $B$. Consider $S = A \cap B$ that is a separator of size $\leq k$. Each component $C$, due to removal of $S$ is a subset of either $A$ or $B$, hence $\dim{\yy_{S \cup C}} \leq |S \cup C|-1 \leq k$. Lemma~\ref{lem:base1} resolves this case for $S = A \cap B$.

If the diameter is $3$, consider a non-leaf edge connecting nodes $X$ and $Y$. $S = X \cap Y$ is separating vertices appearing to the $X$ side of the tree-decomposition, say $A \subseteq V \setminus S$, from vertices appearing to the $Y$ side, say $B \subseteq V \setminus S$. Minimality of $\mathcal{T}$ along with the fact that $X \cap Y$ is a non-leaf edge of $\mathcal{T}$ ensure $|A|, |B| \ge 2$.
Replace $A$ with a single new vertex $a$ at $\yy_a = \mu_{\yy}(A)$ and let $\pi_a = \pi_A$. Connect $a$ to all vertices $X \cap Y$ with equality constraints as current distances $\yy_a-\yy_v, \forall v \in X \cap Y$. The dimension of the new instance can be recursively reduced to $k+1$, with no loss to the variance (that is shifted by a constant if we bring back $A$, considering Lemma~\ref{var-dif}). After the dimension reduction $\yy_{\{a\} \cup (X \cap Y)}$ can be brought back to the original locations (using a unitary transformation, as pairwise distances in between are preserved).
Replace back $A$ and $\yy_a$ in place of $a$ and $\yy_A$. This time replace $B$ with a vertex at the new $\mu_{\yy}(B)$. Reduce the dimension of the embedding to $k+1$ and similarly put back $B$ and $\yy_B$ in place of $b$ and $\mu_{\yy}(B)$. Now we have a separator $S$ of size $\leq k$ with $\dim{\yy_{S \cup A}}, \dim{\yy_{S \cup B}} \leq k+1$. Applying Lemma~\ref{lem:base1} gives the desired.

Finally, consider the case when the diameter is $2$.
Let $X$ be the center of this star $\mathcal{T}$.
If for a leaf $Y$, $|Y \setminus X| \ge 2$ we can use $S = X \cap Y$ to and solve it recursively as in the previous case.
So every leaf $Y$ has a single vertex that is not in $X$. If $|X| \leq k$ the case can be closed by applying Lemma~\ref{lem:base1} for $S = X$.
So assume $|X| = k+1$. Every single vertex from the leaves can be neighbor to at most $k$ vertices in $X$, i.e., is not neighbor to at least one vertex in $X$.

If two such single vertices $v,w$ share a single non-neighbor $u \in X$, let $S = X - u$, and $A = \{v,w\}$. $S$ is separating $A$ from the $B = V \setminus X \setminus A$. If $B$ is singletone, we can apply Lemma~\ref{lem:base1} for $S$ and the three remaining vertices with no edges in between.
Otherwise, $|B| > 1$ and we can perform similar recursion as did for the case of diameter $\ge 3$.

All remains is the case where we have a $k+1$ clique at the center, and (at most) $k+1$ other vertices, each connected to (at most) all but one of the $k+1$ vertices at the center. Let $S$ be the vertices at the center. Applying Lemma~\ref{lem:base1} ensures a solution in $k+2$ dimensions.
A corollary of our (polynomial time) algorithm is the following.

\begin{corollary}
Maximum Variance Embedding, with equality and/or inequality constraints, of an $n$-vertex graph of tree-width $k$ in $\mathbb{R}^n$ has an optimal solution of rank $k+2$.
\end{corollary}

\section{Approximate MVE for General Graphs}
\label{sec:MVE-general}
We showed tree-width$+1$ dimensions is a lower bound for computability of maximum variance embedding. However, this structural bound can be polynomially large, even for sparse graphs such as grids, and as large as  the number of vertices for dense graphs.
As we showed to be the case for trees, where computability is possible in $\mathbb{R}^2$, arbitrary precise approximate solutions in $\mathbb{R}^1$ can be found. For general graphs, the gap is even larger.
In the rest of this Section, we prove the following.

\begin{theorem}\label{thm:JL}
There is a polytime Monte Carlo algorithm that for arbitrary weighted graphs, approximates upper bound constrained maximum variance embedding into $\mathbb{R}^d$  within factor $O((\log n)/d)$ when $d \leq \log n\,,$ and within factor $1 + O\left(\sqrt{(\log n)/d}\right)$ when $d > \log n$. In particular, for $d = \Omega\left((\log n)/\varepsilon^2\right)$, the approximation factor is $1 + \varepsilon$.
\end{theorem}

\begin{proof}
Let $\xx: V \rightarrow \mathbb{R}^n$ be the optimal solution, that can be solved using the full rank SDP. Algorithm~\ref{alg:jl} provides $\yy: V \rightarrow \mathbb{R}^d$, with performance guarantees promised by Theorem~\ref{thm:JL}.

\begin{algorithm} \label{alg:jl}
\caption{Gaussian Rounding}
\begin{algorithmic}
\REQUIRE $\xx: V \rightarrow \mathbb{R}^n$, optimal solution to $\MVED{n}$
\ENSURE $\yy: V \rightarrow \mathbb{R}^d$
\STATE Sample $d$ independent Gaussian vectors $g_i \sim N(0,1)^n$\,, for $i \in [d]$. Let $\GG$ denote $[g_1 g_2 \ldots g_d]^T$ and let $\tau_d = d + 2\sqrt{3 d \log n} + 6 \log n$.
\RETURN $\set{\GG x_u/\sqrt{\tau_d}: u \in V}$.
\end{algorithmic}
\end{algorithm}

For each $u \in V$, let $y_u \defeq \GG x_u$.
Then, for any $u,v \in V$,
\[ \expec{\GG}{\norm{y_u - y_v}^2} = \sum_{i \in [d]} \expec{\GG}{\inprod{g_i, x_u - x_v}^2} = d \norm{x_u - x_v}^2 .\]
Therefore, we get that $\expec{\GG}{\expec{u, v \sim V}{\norm{y_u - y_v}^2}} = d \expec{u, v \sim V}{\norm{x_u - x_v}^2}$.
Using \Cref{fact:gaussian} below, we get that
\begin{equation}
\label{eq:yu-yv}
\Pr\left[ \expec{u,v \sim V}{\norm{y_u - y_v}^2} \geq \frac{d}{2} \expec{u,v \sim V}{\norm{x_u - x_v}^2} \right] \geq \frac{1}{12} .
\end{equation}

\begin{fact}[\cite{LM00}, Lemma 1]
\label{fact:chisquare}
Let $U$ be a $\chi^2$ random variable  with $D$ degrees of freedom. For any positive $t$,
\[ \Pr[U - D \geq 2 \sqrt{D t} + 2 t] \leq e^{-t}. \]
\end{fact}

Since $\frac{\norm{y_u - y_v}^2}{ \norm{x_u - x_v}^2}$ is a $\chi^2$-random variable with $d$ degrees of freedom, plugging in $t = 3 \log n$ in \Cref{fact:chisquare}, we get
\[ \Pr \left[\frac{\norm{y_u - y_v}^2}{\norm{x_u - x_v}^2} - d \geq 2 \sqrt{3 d \log n }+ 6 \log n \right] \leq e^{-3 \log n } = \frac{1}{n^3}. \]

Recall that $\tau_d = d + 2\sqrt{3 d \log n} + 6 \log n$.
Using the union bound over all pairs of vertices in $V$ we get
\[ \Pr \left[\frac{\norm{y_u - y_v}^2}{\norm{x_u - x_v}^2}  \leq \tau_d \ \forall u,v \in V \right] \geq 1 - \frac{1}{n}. \]
Hence, w.h.p., for $\set{u,v} \in E$,
$ {\norm{y_u - y_v}^2}/{\tau_d} \leq \norm{x_u - x_v}^2 \leq 1\,.$
Therefore, w.h.p., $\set{y_u/\sqrt{\tau_d}: u \in V}$ is a $d$-dimensional Lipschitz embedding of $G$.
From \Cref{eq:yu-yv}, we get that
\[  \Pr\left[ \expec{u,v \sim V}{\frac{\norm{y_u - y_v}^2}{\tau_d}} \geq \frac{d}{2\tau_d} \expec{u,v \sim V}{\norm{x_u - x_v}^2} \right] \geq \frac{1}{12} . \]

Using the union bound over these two events, we get a $O(\tau_d/d)$ approximation to MVE in $\mathbb{R}^d$, with constant probability. Note that
\[ \tau_d/d = \begin{cases} O((\log n)/d) & d \leq \log n \\
    1 + O\left(\sqrt{(\log n)/d} \right) & d > \log n \end{cases} . \]
In particular, when $d = \Omega((\log n)/\varepsilon^2)$, $\tau_d/d = 1 + \varepsilon$.

\begin{fact}[Folklore]
\label{fact:gaussian}
Let $g_1, \ldots, g_l$ be (not necessarily independent) Gaussian random variables each having mean $0$, and $\expec{}{g_i^2} = \sigma_i^2$.
Then,
\[ \Pr \left[ \sum_{i \in [l]} g_i^2 \geq \frac{1}{2} \sum_{i \in [l]} \sigma_i^2  \right] \geq \frac{1}{12}.\]
\end{fact}
\end{proof}

\begin{proof}
\begin{align*}
\expec{}{\left(\sum_{i \in [l]} g_i^2 \right)^2} & = \sum_{i \in [l]} \expec{}{g_i^4} + 2 \sum_{\substack{i,j \in [l]\\ i \neq j}} \expec{}{g_i^2 g_j^2} \\
 & \leq 3 \sum_{i \in [l]}\sigma_i^4 + 2 \sum_{\substack{i,j \in [l]\\ i \neq j}} \sqrt{\expec{}{g_i^4}} \sqrt{\expec{}{g_j^4}} & \textrm{(Cauchy-Schwarz inequality)} \\
 & = 3 \left( \sum_{i \in [l]}\sigma_i^4 + 2 \sum_{\substack{i,j \in [l]\\ i \neq j}} \sigma_i^2 \sigma_j^2 \right)
 = 3 \left( \sum_{i \in [l]}\sigma_i^2\right)^2.
\end{align*}
Using the Paley-Zygmund inequality,
\[ \Pr \left[ \sum_{i \in [l]} g_i^2 \geq \frac{1}{2} \sum_{i \in [l]} \sigma^2  \right] \geq \left(1 - \frac{1}{2} \right)^2 \cdot \frac{\left( \sum_{i \in [l]}\sigma_i^2\right)^2}{3 \left( \sum_{i \in [l]}\sigma_i^2\right)^2} = \frac{1}{12} . \]
\end{proof}

\section{Conclusion}

In Sections~\ref{sec:lambda-infinity} and \ref{sec:MVE-tree} we proved that computing $\lambda_\infty$ and spread constant is NP-hard even for simplest graphs and proposed efficient approximation algorithms in this case.
The gap between best known algorithms, e.g., that of Section~\ref{sec:MVE-general}, and hardness results remain enormous and encourage further study and investigation of these problems, from both sides. Moreover, breaking barriers such as Small Set Expansion hardness by approximating $\lambda_\infty$ better than a factor of $O(\log \Delta)$ \cite{LRV13, ST12, RST12, RS10} further encourages study of such functional graph parameters.

In Section~\ref{sec:MVE-tree-width} we showed that higher dimensional relaxations of spread constant, and more generally Maximum Variance Embedding, remain NP-hard up to tree-width dimensions.
Our combinatorial and non-convex optimization algorithms from Sections~\ref{sec:MVE-tree} and \ref{sec:MVE-tree-width} imply that this bound is (computationally) tight for graphs of constant tree-width. However tree-width itself is NP-hard to compute (for the literature, refer to \cite{korhonen2022single} and references therein) and state-of-the-art algorithms (e.g., \cite{feige2008improved}) allow us to compute MVE in polynomial time, down to $O(w \sqrt{\log w})$ dimensions where $w$ is tree-width of the graph.
Closing the gap between $w$ for NP-hardness and $O(w \sqrt{\log w})$ for polytime algorithms, is yet another interesting open problem.

\newpage

\bibliography{refs}
\bibliographystyle{amsalpha}

\newpage

\appendix

\section{Omitted Proofs}
\label{sec:omittedproofs}

\subsection{$\lambda_\infty$ is an Algebraic Number}

\begin{lemma}[Folklore]
\label{lem:infmin}
There is at least one optimal valuation $f: V \rightarrow \mathbb{R}$, achieving $\lambda_\infty$ in equation  \ref{def:lambdainfty}.
\end{lemma}

\begin{proof}
From the definition we have
\begin{align*}
\lambda_\infty &= \inf_{f: V \rightarrow \mathbb{R}} \frac{\int_V \sup_{y: y \in N(x)} |f(x)-f(y)|^2 \diff \pi(x)}{\var[v \sim \pi]{f(v)} } && \text{Equation \ref{def:lambdainfty}} \\
&= \inf_{x \in \mathbb{R}^ \setminus \{\bm{0}\}} \frac{\expec{v \sim \pi}{\max_{u \in N(v)}  |x_u-x_v|^2} }{\var[v \sim \pi]{x_v} } \\
&= \inf_{x \in \mathbb{R}^ \setminus \{\bm{0}\}} \frac{\expec{v \sim \pi}{\max_{u \in N(v)}  |x_u-x_v|^2} }{\expec{v \sim \pi}{|x_v- \expec{u \sim \pi}{x_u}|^2} } \\
&= \inf_{x \in \mathbb{R}^ \setminus \{\bm{0}\}} \frac{\expec{v \sim \pi}{\max_{u \in N(v)}  |(x_u+\nu)-(x_v+\nu)|^2} }{\expec{v \sim \pi}{|(x_v+\nu)- \expec{u \sim \pi}{(x_u+\nu)}|^2} } &&  \text{shift invariance } \forall \nu \in \mathbb{R} \\
&= \inf_{x \in \mathbb{R}^ \setminus \{\bm{0}\}, \expec{u \sim \pi}{x_u} = 0} \frac{\expec{v \sim \pi}{\max_{u \in N(v)}  |x_u-x_v|^2} }{\expec{v \sim \pi}{x_v^2} } && \text{let $\nu = -\expec{u \sim \pi}{x_u} = 0$} \\
 &= \inf_{x \in \mathbb{R}^ \setminus \{\bm{0}\}, \expec{u \sim \pi}{x_u} = 0} \frac{\expec{v \sim \pi}{\max_{u \in N(v)}  |\rho x_u-\rho x_v|^2} }{\rho^2 \expec{v \sim \pi}{x_v^2} } && \forall \rho > 0  \\
&= \inf_{x \in \mathbb{R}^ \setminus \{\bm{0}\}, \expec{u \sim \pi}{x_u} = 0, \expec{u \sim \pi}{x_u^2} = 1} {\expec{v \sim \pi}{\max_{u \in N(v)}  | x_u- x_v|^2} } && \text{for $\rho^2 = \left(\expec{u \sim \pi}{x_u^2}\right)^{-1} $}.
\end{align*}
Therefore,
\begin{equation}
\label{beforecompose}
\lambda_\infty = \inf_{x^T \pi = 0, x^T \text{diag}(\pi) x = 1} {\expec{v \sim \pi}{\max_{u \in N(v)}  | x_u- x_v|^2} }\,. 
\end{equation}

Considering a pair of permutations
$p : V \rightarrow [n]$ and $q : E \rightarrow [m]\,,$ where $m$ is the number of edges,
define the set
\[S_{(p,q)} \defeq \{ x \in \mathbb{R}^n : \,  x_u \leq x_v \,, \forall \ p(u) < p(v), \ 
 |x_v - x_u| \leq |x_{v'} - x_{u'}|\,, \forall \ q(u,v) < q(u', v') \}. \]
Note that $\text{sign}(x_u-x_v)$ is uniform across $S_{(p,q)}$ for every pair of vertices $u, v \in V$. Therefore every $S_{(p,q)}$ can be redefined without absolute values in the notation and hence is a polyhedron. Moreover, $\cup_{p,q} S_{(p,q)} = \mathbb{R}^n$.
Finally, due to the constraints corresponding to permutation of edges, in the definition of $S_{(p,q)}\,,$ the ``farthest'' neighbor $v$ to each vertex $u$ is constant for all $x \in S_{(p,q)}$. In other words, for every vertex $u \in V$ there exists a vertex $v \in V$ such that 
$$
|x_v - x_u| = \max_{w \in N(u)} |x_w - x_u| \quad \forall x \in S_{(p,q)}.
$$ 
We denote this $v$ by $f_{p,q,u}$.

Since, $\mathbb{R}^n = \cup_{(p,q) \in [n!] \times [m!]} S_{(p,q)}$, we can rewrite $\lambda_\infty$ from equation \ref{beforecompose} as
\begin{equation}
\lambda_\infty = \min_{(p,q) \in [n!] \times [m!]} \inf_{x \in S_{(p,q)}, x^T \pi = 0, x^T \text{diag}(\pi) x = 1} {\expec{v \sim \pi}{  ( x_u- x_{f_{p,q,u}})^2} }\,. \label{composedinf}
\end{equation}
For a fixed $S_{(p,q)}$, the function under infimum is a degree $2$ polynomial with coefficients from $\mathbb{Z}[\pi]$,
$${\expec{v \sim \pi}{  ( x_u- x_{f_{p,q,u}})^2} } = \sum_u \pi_v  (x_u- x_{f_{p,q,u}})^2\,,$$ hence is continuous.

Moreover, for each case $(p,q)$ the domain of infimum is the intersection of surface of an ellipsoid
$$
x^T \text{diag}(\pi) x = 1\,,
$$
a hyperplane
$$
x^T \pi = 0\,,
$$
and a polyhedron $S_{(p,q)}\,.$ Since all three are compact sets, their intersection is also a compact set.
Therefore, each $\inf$ in equation \ref{composedinf} computes the infimum of a continuous function over a compact set. Hence, applying the {\em Extreme Value Theorem}, we can replace these infima with $\min\,,$
\[
\lambda_\infty = \min_{(p,q) \in [n!] \times [m!]} \min_{x \in P_{(p,q)}, x^T \pi = 0, x^T \text{diag}(\pi) x = 1} {\expec{v \sim \pi}{  ( x_u- x_{f_{p,q,u}})^2} }, \]
allowing us to replace $\inf$ with $\min$ in either $\lambda_\infty$ formulation.
\end{proof}

\begin{corollary}
\label{nearrational}
$\lambda_\infty$ is an algebraic number for rational inputs.
\end{corollary} 

\begin{proposition}
Spread constant is a rational number, if not $\infty$, for rational inputs.
\end{proposition}
\begin{proof}
Following notation from Lemma~\ref{spread_constant_is_binary}, Without loss of generality, we may assume $y_i = 0$ for some arbitrary vertex $i\,.$ Lipschitz conditions give constraints of form
$$
y_u - y_v \le 1, y_v - y_u \le 1, \forall (u,v) \in E\,.
$$
that defines a bounded polytope in $\mathbb{R}^V$ assuming the graph is connected. Vertices of the polytope, being intersection of $n$ constraints holding at equality (hyperplanes), are rational.

Variance, being a linear combination of functions of form
$$(y_i - y_j)^2\,,$$ is strongly convex so is maximized at a vertex of our polytope and has rational value there,  which equals the spread constant of the graph.
\end{proof}

\subsection{NP Hardness for Spread Constant of a Star}\label{ssec:spread-constant-NPH}

In this Section we prove Theorem~\ref{thm:sc-nph_theorem}. Given a {\sc Partition} instance $P = \set{p_1, \ldots, p_{n-1}}$ define a star graph $S_n$ and a distribution
$
\pi_0 = 1 -  \beta \,,
\pi_i =  \beta \frac{p_i}{\sum_j p_j} ~ \forall i \ge 1\,.
$
The proof is immediate given the following lemma.

\begin{lemma}
The answer to {\sc Partition} of $P$ is {\sc yes} if and only if $\spcon = \beta\,,$ otherwise $\spcon < \beta - \Omega_\beta(\frac{1}{(\sum_i p_i)^2})$.
\end{lemma}

\begin{proof}[Proof sketch]
Let $y$ be an optimum solution to $\spcon(S_n,\pi)$. I.e., $y: V \rightarrow \mathbb{R}$ is a Lipschitz valuation maximizing $\var[\pi]{y}$.
A direct consequence of Lemma~\ref{full-stretch} is the following binary characterization of the optimum.

\begin{lemma}
\label{spread_constant_is_binary}
All leaf values are at unit distance from the root, i.e.,
$
|y_j-y_0| = 1\,, \forall j \in [n-1]\,.
$
\end{lemma}

Assuming $y_0 = 0$ (due to shift invariance of the problem) let
$y_i \in \set{\pm1}, \forall i \in [n-1]$.

Defining $\pi_- = \sum_{i} \pi_i \cdot \mathds{1}[y_i = -1]$ and
$
\pi_+ = \sum_{i} \pi_i \cdot \mathds{1}[y_i =  +1]\,,
$
we can write
\begin{align}
\spcon &= \var[\pi]{y} = \sum_{j < k} \pi_j \pi_k (y_j-y_k)^2 \nonumber \\
&= 4 \pi_- \pi_+ +  (\pi_- + \pi_+) \pi_0 
= 4 \pi_- \pi_+ +  (1 - \pi_0) \pi_0  \nonumber \\
&\leq 4 \left(\frac{\pi_- + \pi_+}{2}\right)^2 +  (1 - \pi_0) \pi_0 \label{varsplit} && \text{AM-GM inequality}\\
&=   (1 - \pi_0)^2 +  (1 - \pi_0) \pi_0
= {1 - \pi_0}\,. \nonumber
\end{align}
The above shows $\spcon \leq {1 - \pi_0} = \beta$ and moreover holding with equality if and only if inequality \eqref{varsplit} does; i.e., when $\pi_- = \pi_+$. The case of {\sc yes} is immediate.
If the answer to {\sc Partition} is {\sc no}, given $|\pi_- - \pi_+| > \beta \frac{1}{\sum_j p_j}$ one can upper bound $\spcon$ by $\beta - \Omega_\beta(\frac{1}{(\sum_j p_j)^2})$ using  \eqref{varsplit}.
\end{proof}

\section{Applying Grothendieck Inequality for MVE of Trees}\label{sec:MVE-tree-Grothendieck}

Corresponding to an embedding $\{\yy_i : i \in [n]\}$ consider $\zz_e \defeq \yy_i-\yy_j \quad \forall e = (i,j) \in E$.
For trees, Lipschitzness constraints will be equivalent to
$
\|z_e\|^2 = \langle \zz_e\,\zz_e \rangle \leq \ell_{e}^2\, \quad \forall e \in E\,
$
where we have $|E| = n-1$ variables $\{\zz_e \in \mathbb{R}^d : e \in E\}$. Finally, let $\mathbb{R}^{|E| \times d} \ni \ZZ = [\zz_e : e \in E]^T$\,.

Let $P_{ij}$ denote the unique path between vertices $i$ and $j$. We can reformulate the variance as 
\begin{align}
    \var[]{\yy} 
    = \frac{1}{2} \sum_{i \in V} \sum_{j \in V} \left\|\yy_i - \yy_j \right\|^2 \pi_i \pi_j = \frac{1}{2} \sum_{i \in V} \sum_{j \in V} \| \sum_{e \in P_{ij}} \zz_e \|^2 \pi_i \pi_j = \langle \AA , \ZZ \ZZ^T \rangle\,
\end{align}
where $\AA \in \mathbb{R}^{|E| \times |E|}$ is defined as 
$
\AA_{e e'} = \frac{1}{2} \sum_{i \in V} \sum_{j \in V} \mathbbm{1}[\set{e,e'} \subseteq P_{ij}] \pi_i \pi_j
$.

Defining $\XX = \ZZ \ZZ^T$, MVE in $\mathbb{R}^k$ can be formulated as the following SDP.
\begin{equation*}
\begin{array}{ll@{}ll}
\text{maximize}  & & \langle \AA, \XX \rangle &\\
\text{subject to}
&   & \XX_{ee} \leq \ell_{e}^2 ,  & \forall\quad e \in E\\
 &  & \XX \succcurlyeq \zero_{|E| \times |E|}, \quad \text{rank}(\XX) \leq k & 
\end{array}
\end{equation*}

Applying Theorem~1 from \cite{MMMO17}, in the case $\ell_e = 1 ~ \forall e \in E$ guarantees $\frac{1}{k-1}$ optimality of local optima for lifting to $k$ dimensions, i.e., \Cref{thm:main1}.

\end{document}